\documentclass[pra,aps,twocolumn,notitlepage,superscriptaddress,nofootinbib]{revtex4-2}

\usepackage{amsmath,amssymb,amsthm,bm}
\usepackage{graphicx}
\usepackage{tikz}
\usetikzlibrary{arrows.meta,positioning,fit,calc,backgrounds,matrix,patterns}
\usepackage{algpseudocode}
\usepackage{enumitem}
\usepackage{xcolor}
\usepackage[hidelinks]{hyperref}

\newcommand{\F}{\mathbb F}
\newcommand{\C}{\mathbb C}

\newcommand{\Sym}{\operatorname{Sym}}
\newcommand{\rank}{\operatorname{rank}}
\newcommand{\Tr}{\operatorname{Tr}}
\newcommand{\im}{\operatorname{Im}}
\newcommand{\ket}[1]{\lvert #1\rangle}
\newcommand{\bra}[1]{\langle #1\rvert}
\newcommand{\braket}[2]{\langle #1\mid #2\rangle}

\newtheorem{definition}{Definition}
\newtheorem{theorem}[definition]{Theorem}
\newtheorem{lemma}[definition]{Lemma}
\newtheorem{proposition}[definition]{Proposition}
\newtheorem{corollary}[definition]{Corollary}
\newtheorem{remark}[definition]{Remark}

\newcounter{protocol}
\renewcommand{\theprotocol}{\arabic{protocol}}
\newenvironment{protocol}[1]{%
  \refstepcounter{protocol}%
  \par\medskip
  \noindent\begin{minipage}{\columnwidth}%
  \hrule\smallskip
  \textbf{Protocol \theprotocol: #1}\par\smallskip
}{%
  \smallskip\hrule
  \end{minipage}%
  \par\medskip
}

\begin{document}

\title{Stabilizer-Public-Key Authentication:\\
Long Stabilizer Public Keys Resist Finite-Copy Forgery Attacks}

\author{Masahito Hayashi}
\thanks{These authors contributed equally to this work.}
\email{hmasahito@cuhk.edu.cn}
\affiliation{School of Data Science, The Chinese University of Hong Kong, Shenzhen, Longgang District, Shenzhen, 518172, China}
\affiliation{International Quantum Academy, Futian District, Shenzhen 518048, China}
\affiliation{Graduate School of Mathematics, Nagoya University, Nagoya, 464-8602, Japan}
\author{Jingtian Zhao}
\affiliation{School of Data Science, The Chinese University of Hong Kong, Shenzhen, Longgang District, Shenzhen, 518172, China}
\author{Baichu Yu}
\affiliation{Shenzhen Institute for Quantum Science and Engineering,
Southern University of Science and Technology, Nanshan District, Shenzhen, 518055, China}
\affiliation{Quantum Science Center of Guangdong-Hong Kong-Macao Greater Bay Area, Shenzhen 518045, China}

\begin{abstract}
We propose an information-theoretic authentication protocol based on a finite supply of long quadratic-stabilizer public-key states over an odd-prime field, where the effective key length after one signature exposure is the residual dimension $r=n-\ell$. A computationally unbounded adversary observes one valid classical signature and may jointly process $N$ public-key copies, while verification uses one additional independent copy. We show that, conditioned on the exposed signature, the security problem reduces to a partial-prediction game for a uniform stabilizer ensemble on these $r$ residual qudits, with a partial query along $d=\operatorname{rank}(Y-Y')$ directions, where $Y$ and $Y'$ denote the honestly signed and target-forged messages, respectively. Surprisingly, although the target requires only partial information, there is no first-order reduction in the required copy rate when $d/r\to\beta\in(0,1]$. The optimal average forgery probability for the specified target tends to zero for $N/r\to\alpha<1$ and to one for $\alpha>1$, revealing a sharp threshold at $\alpha=1$. Thus, long stabilizer public keys resist finite-copy forgery attacks under a single-signature-exposure model whenever the adversarial copy rate remains below one.
\end{abstract}

\maketitle

\section{Introduction}
\label{sec:introduction}

Digital signatures provide public verifiability of message origin and
integrity, and in classical cryptography their existence is closely tied to
computational one-wayness~\cite{Rompel1990}.  Large-scale quantum computers
would invalidate widely deployed number-theoretic assumptions through Shor's
algorithms for factoring and discrete logarithms~\cite{Shor1994}.  The primary
practical response is post-quantum cryptography, including the standardized
ML-DSA and SLH-DSA families~\cite{BernsteinLange2017,NISTFIPS204,NISTFIPS205}.
These developments also sharpen a distinct question: can a public key that is
itself a physical quantum state support an information-theoretic
authentication guarantee when only finitely many copies are available?

The quantum-public-key paradigm was articulated by Gottesman and
Chuang~\cite{GottesmanChuang2001}.  Here a ``public key'' is a provisioned
verification resource: a bounded number of authenticated physical copies may
be distributed, rather than an unrestricted classical string that anyone can
copy.  No-cloning alone does not provide unforgeability.  The number of copies
must be treated as a cryptographic resource, because sufficiently many copies
may reveal the state whereas too few may leave an irreducible prediction
error.

This setting is operationally different from both communication-oriented
quantum digital signatures and computational signatures involving quantum
objects.  Multipartite QDS protocols emphasize transferability,
non-repudiation, and network implementation
\cite{Puthoor2016,Roberts2017,Richter2021,Yin2023,Du2025}.  Computational
quantum-signature research instead restricts the adversary's running time and
can pursue standard multi-use notions
\cite{MorimaeYamakawa2022,MorimaeYamakawa2024,KhuranaTomer2024,
ColadangeloMutreja2024,MorimaePorembaYamakawa2024}.  The present work studies a
single verifier using a locally held quantum public-key state and allows a
computationally unbounded adversary, but explicitly limits the adversary's
physical supply of that state.  A recent single-qubit finite-copy protocol
provides the closest direct comparison~\cite{WangHayashi2026QDS}; detailed
comparisons with these neighboring models are deferred to
Section~\ref{sec:discussion}.

Our secret key is a classical pair $(X,b)$ specifying a quadratic stabilizer
state over an odd-prime field, and one copy of that state is the quantum public
key.  A message selects an $\ell$-dimensional commuting family of generalized
Pauli observables.  Its classical signature reports the corresponding
restrictions of the secret matrix and joint-eigenvalue assignment, and a
verifier tests them on an independently allocated public-key copy.  We analyze
a computationally unbounded adversary that receives one honestly generated
message--signature pair, a distinct target specified independently of the hidden key and exposed signature, and $N$ public-key copies and
may perform an arbitrary collective quantum operation.

The central difficulty is the disclosed valid signature.  It cannot be
treated as independent side information, and forgery for the specified target cannot be
reduced directly to reconstruction of the complete secret key.  We show
instead that, after conditioning on the exposed signature and applying a known
coordinate change, every public-key copy separates into a known factor and a
uniform quadratic stabilizer state on
\begin{equation}
  r=n-\ell
  \label{eq:introduction-residual-dimension-meaning}
\end{equation}
residual qudits.  Thus $r$ is the dimension of the residual stabilizer ensemble
that remains hidden after one signature has been disclosed.  The specified target
message $Y'$ introduces
\begin{equation}
  d=\rank(Y-Y'),
  \label{eq:introduction-target-rank}
\end{equation}
which is exactly the number of new residual directions queried by that target.

This distinction leads to the main surprise. When $d<r$, the target asks for only part of the residual information, so one might expect the required number of public-key copies to decrease in proportion to $d/r$. However, no such first-order reduction occurs: if $d/r\to\beta$ for any fixed $\beta\in(0,1]$, the optimal average forgery probability tends to zero when $N/r\to\alpha<1$ and to one when $\alpha>1$.

This establishes that long stabilizer public keys resist finite-copy forgery attacks. More precisely, as the hidden residual dimension $r=n-\ell$ tends to infinity, the forgery probability tends to zero when the adversary has $N=\alpha r+o(r)$ public-key copies with $\alpha<1$. The number $N$ of available copies may therefore tend to infinity, provided that it grows at a rate strictly below the residual dimension. The critical line $N/r\to1$ and the sublinear regime $d=o(r)$ remain unresolved.

The proof combines two ingredients.  This paper proves the
authentication-specific conditional reduction and shows that arbitrary
well-formed forged outputs are upper-bounded by the relevant residual
partial-prediction game.  The independent work \emph{No Copy-Rate Discount for
Partial Stabilizer Learning} supplies the learning-theoretic no-discount
result for a broader stabilizer-state setting
\cite{HayashiQuadraticStabilizerPrediction}.  Above the threshold, complete
identification of the residual label allows the adversary to reconstruct an
honest target signature.  The formal parameter correspondence and the exact
division of proof responsibilities are recorded in
Section~\ref{subsec:parameter-correspondence}.

Finite quantum-copy security does not permit unrestricted reuse of the
classical signing key.  Repeated signatures reveal the secret parameters on
the cumulative span of their message subspaces and can eventually determine
the complete key.  The result should therefore be read as an authentication theorem for one exposed
signature and a separately specified target, not as a conventional reusable
signature guarantee.  Section~\ref{sec:contributions-scope} states the main
contributions and exact security scope before the technical development.

The remainder of the paper is organized as follows.
Section~\ref{sec:preliminaries} fixes the finite-field, Weyl, and stabilizer
conventions.  Sections~\ref{sec:protocol} and
\ref{sec:correctness-resources} define the protocol, verification measurement,
and public-key copy accounting.  Section~\ref{sec:security-model} defines the
forgery for the specified target experiment.  Section~\ref{sec:conditional-reduction}
proves the conditional residual ensemble and the one-directional forgery
reduction.  Section~\ref{sec:finite-copy-security} identifies the residual game
with the independent learning problem and determines the corresponding forgery copy
rate.  Section~\ref{sec:key-reuse} treats multiple-signature leakage and key
recovery.  Section~\ref{sec:discussion} discusses neighboring models,
limitations, and open problems.

\section{Contributions, relation to stabilizer learning, and security scope}
\label{sec:contributions-scope}

\subsection{Main contributions}
\label{subsec:introduction-contributions}
The main contributions are as follows.
\begin{enumerate}[leftmargin=*]
  \item We give a self-contained stabilizer-public-key authentication protocol
        and a single-signature, security for a specified target experiment that permits
        arbitrary collective attacks on $N$ public-key copies while reserving
        an independent copy for verification.
  \item We prove an exact conditional residual-state reduction.  One exposed
        signature leaves a uniform quadratic stabilizer ensemble on
        $r=n-\ell$ residual qudits, and a target at rank distance
        $d=\rank(Y-Y')$ queries exactly $d$ new residual directions.
  \item We reduce arbitrary well-formed forged verification effects to the
        corresponding residual graded prediction-and-verification game.  The reduction is
        one-directional and uses compression of the full commuting projector
        together with commutativity-preserving randomized completion.
  \item Combining this reduction with the independent partial-stabilizer-
        learning theorem, we prove a copy-rate-one transition for every fixed
        positive query fraction $d/r\to\beta\in(0,1]$.  Below $N/r=1$ the
        optimal average forgery probability for the specified target vanishes; above it,
        complete residual identification permits reconstruction of an honest
        target signature.
  \item We characterize the separate classical leakage caused by repeated
        signatures, prove an exact span-forgery criterion, and show that full
        observed rank recovers the complete secret key.
\end{enumerate}

\subsection{Parameter correspondence and division of proof responsibilities}
\label{subsec:parameter-correspondence}
After one valid signature is disclosed, $r=n-\ell$ is the dimension of the
residual stabilizer ensemble that remains hidden in every public-key copy.  A
the specified target $Y'$ introduces $d=\rank(Y-Y')$ new residual directions, and $N$
is the number of residual copies available to the adversary.  These play the
roles of ambient dimension, query dimension, and learner-copy count in the
partial-stabilizer-learning problem:
\begin{equation}
  (n,m,k)_{\rm learn}=(r,d,N)_{\rm auth}.
  \label{eq:introduction-learning-crosswalk}
\end{equation}
Here, in the notation of the partial-stabilizer-learning problem, $n$ is the ambient number of qudits, $m$ is the dimension of the queried subspace, and $k$ is the number of copies available to the learner; the verifier's independent test copy is not included in $k$.
The distinction between $r$ and $d$ is essential.  When $d<r$, a forger need
not reconstruct the complete residual label.  Nevertheless, if
$d/r\to\beta$ for any fixed $\beta\in(0,1]$, restricting the target to those
$d$ directions does not lower the first-order copy threshold below $N/r=1$.
The finite-size optimum and its decay estimates still depend on $d$ and
$\beta$; only the normalized first-order threshold is independent of the
fixed positive query fraction.  The statement therefore neither asserts
identical finite-size security for all $d$ nor covers $d=o(r)$.

The no-copy-rate-discount theorem is proved in the independent work, which treats a
broader stabilizer-learning problem, including general stabilizer-state
families~\cite{HayashiQuadraticStabilizerPrediction}. 
The present paper proves
the authentication-specific steps: conditioning on the exposed signature,
identifying $d$ with the number of new residual directions, upper-bounding
arbitrary forged verification effects by the residual graded game, and
converting complete residual identification into an honest target signature.
The authentication optimum is only upper-bounded by the residual game; the
residual game itself is exactly the symmetric-matrix stabilizer
graded-verification problem analyzed in the independent learning work.

\subsection{Exact scope of the security guarantee}
\label{subsec:security-scope}
The threshold proved here is deliberately narrower than standard signature
unforgeability.  It assumes an odd prime $p$, a uniform secret-key prior,
ideal state preparation and measurement, authenticated distribution of intact
public-key copies, one exposed valid signature, and one distinct target specified independently of the hidden key and exposed signature.
The adversary is computationally unbounded and may process all $N$ supplied
copies collectively, but the final verifier uses one additional independent
copy.  The probability is averaged over key generation.  None of these
qualifications is removed by the copy-rate theorem.

\begin{proposition}[Exact scope of the proved guarantee]
  \label{prop:security-scope-summary}
  Fix an odd prime $p$.  Consider protocol instances with residual dimension
  $r=n-\ell$, one honestly signed message $Y$, one specified target $Y'\ne Y$, and
  target rank distance
  \begin{equation}
    d=\rank(Y-Y').
    \label{eq:scope-target-rank}
  \end{equation}
  Under ideal preparation and measurement, authenticated distribution, an
  independent verifier copy, and the uniform key prior, the following
  statements hold for the optimized average forgery probability for the specified target.
  \begin{enumerate}[label=(\roman*),leftmargin=*]
    \item If $d/r\to\beta\in(0,1]$ and $N/r\to\alpha<1$, then
          $P_{\mathrm{forge}}^{(N),*}(Y\to Y')\to0$.
    \item For every sequence with specified targets, $N-r\to+\infty$ implies
          $P_{\mathrm{forge}}^{(N),*}(Y\to Y')\to1$.  In particular, this holds
          whenever $N/r\to\alpha>1$.
    \item Consequently, for every $0<\beta\le1$ and $0<\epsilon<1$, the
          average forgery copy rate for the specified target defined in
          Eq.~\eqref{eq:authentication-forgery-rate-definition} satisfies
          \begin{equation}
            R_{\mathrm{auth,forge}}(p,\beta,\epsilon)=1.
            \label{eq:scope-copy-rate-one}
          \end{equation}
    \item No pointwise conclusion is made when $N/r\to1$, and the subcritical
          statement does not cover $d=o(r)$.
  \end{enumerate}
\end{proposition}

The proposition concerns a target selected independently of the hidden
key and exposed signature.  It does not imply adaptive-target security,
existential unforgeability, or worst-case-key security.  One exposed signature
is not a chosen-message experiment: Section~\ref{sec:key-reuse} shows that
several signatures leak the secret parameters on their cumulative input span
and can eventually reveal the complete key.  The theorem therefore does not
establish EUF-CMA security or unrestricted secret-key reuse.

The analysis also assumes authenticated public-key distribution, ideal
operations, and a fresh verifier copy.  It does not establish non-repudiation,
transferability, noise tolerance, safe reuse of a verification copy, or
composable security.  These exclusions remain in force on both sides of the
copy-rate threshold.

\section{Quadratic stabilizer preliminaries}
\label{sec:preliminaries}

This section fixes the algebraic and phase conventions used throughout the
paper.  We use the same phase-adjusted Weyl convention as the independent learning work
on multiple-copy prediction of quadratic stabilizer measurement outcomes
\cite{HayashiQuadraticStabilizerPrediction}.  This alignment is essential for
the later reduction: the secret linear functional in the authentication
protocol will be identified explicitly with the character label of the
quadratic stabilizer state.

\subsection{Finite-field and symplectic notation}
\label{subsec:finite-field-symplectic}

Let $p$ be an odd prime, let $\F_p$ be the field with $p$ elements, and set
\begin{equation}
  \omega:=e^{2\pi i/p}.
  \label{eq:omega}
\end{equation}
All vector and matrix operations in exponents are evaluated over $\F_p$.
In particular, $1/2$ denotes the multiplicative inverse of $2$ in $\F_p$.
The Hilbert space of $n$ qudits is
\begin{equation}
  \mathcal H_n:=(\C^p)^{\otimes n},
  \label{eq:Hilbert-space}
\end{equation}
with computational basis $\{\ket{x}:x\in\F_p^n\}$.

We write a phase-space vector as $(u,z)\in\F_p^n\oplus\F_p^n$ and use the
symplectic form
\begin{equation}
  [(u,z),(u',z')]
  :=z^T u'-u^T z'.
  \label{eq:symplectic-form}
\end{equation}
For a subspace $L\subseteq\F_p^{2n}$, its symplectic orthogonal complement is
\begin{equation}
  L^{\perp_s}
  :=\{g\in\F_p^{2n}:[g,h]=0\ \text{for every }h\in L\}.
  \label{eq:symplectic-orthogonal}
\end{equation}
An $n$-dimensional subspace $L$ satisfying $L=L^{\perp_s}$ is called
Lagrangian, meaning that it is maximal among subspaces whose vectors are
mutually symplectically orthogonal.  Equivalently for the present purpose, it
indexes a maximal commuting family of Weyl operators.

\subsection{Phase-adjusted Weyl operators}
\label{subsec:Weyl-operators}

For $u,z\in\F_p^n$, define the shift and phase operators by
\begin{equation}
  \mathsf X(u)\ket{x}:=\ket{x+u},
  \qquad
  \mathsf Z(z)\ket{x}:=\omega^{z^Tx}\ket{x}.
  \label{eq:shift-phase}
\end{equation}
They obey
\begin{equation}
  \mathsf X(u)\mathsf Z(z)
  =\omega^{-u^Tz}\mathsf Z(z)\mathsf X(u).
  \label{eq:XZ-commutation}
\end{equation}
We use the phase-adjusted Weyl operator
\begin{equation}
  W(u,z)
  :=\omega^{-\frac12u^Tz}\mathsf Z(z)\mathsf X(u).
  \label{eq:Weyl-definition}
\end{equation}
The multiplication and commutation rules are
\begin{align}
  W(u,z)W(u',z')
  &=\omega^{\frac12(z^Tu'-u^Tz')}
    W(u+u',z+z'),
  \label{eq:Weyl-product}\\
  W(u,z)W(u',z')
  &=\omega^{[(u,z),(u',z')]}
    W(u',z')W(u,z).
  \label{eq:Weyl-commutation}
\end{align}
Consequently, two Weyl operators commute exactly when their phase-space
labels are symplectically orthogonal.

\subsection{Symmetric-matrix graph charts}
\label{subsec:graph-charts}

Let $X\in\F_p^{n\times n}$.  Its graph subspace, called the graph chart determined by $X$, is
\begin{equation}
  L_X
  :=\left\{\binom{u}{Xu}:u\in\F_p^n\right\}
  \subseteq\F_p^{2n}.
  \label{eq:graph-subspace}
\end{equation}
For $u,u'\in\F_p^n$,
\begin{equation}
  \left[\binom{u}{Xu},\binom{u'}{Xu'}\right]
  =u^T(X^T-X)u'.
  \label{eq:graph-isotropy-calculation}
\end{equation}
Hence $L_X$ is Lagrangian if and only if $X=X^T$.  Throughout this paper,
$X$ is sampled from
\begin{equation}
  \Sym_n(\F_p)
  :=\{X\in\F_p^{n\times n}:X=X^T\}.
  \label{eq:symmetric-matrix-space}
\end{equation}
For symmetric $X$, Eq.~\eqref{eq:Weyl-product} has no residual phase on
$L_X$, so
\begin{equation}
  W(u,Xu)W(u',Xu')
  =W(u+u',X(u+u')).
  \label{eq:graph-Weyl-representation}
\end{equation}
Thus $u\mapsto W(u,Xu)$ is an ordinary unitary representation of the additive
group $\F_p^n$.

\subsection{Characters and quadratic stabilizer states}
\label{subsec:quadratic-stabilizer-states}

A linear functional on $L_X$ is uniquely specified by its pullback along the
graph isomorphism $u\mapsto(u,Xu)$.  Thus, for every $v\in L_X^*$ there is a
unique row vector $a^T\in(\F_p^n)^*$ such that
\begin{equation}
  v\!\left(\binom{u}{Xu}\right)=a^Tu
  \qquad (u\in\F_p^n).
  \label{eq:functional-pullback}
\end{equation}
To match the convention of the independent prediction paper, we parameterize
$a=-b$ and write
\begin{equation}
  v_{X,b}\!\left(\binom{u}{Xu}\right):=-b^Tu,
  \qquad b\in\F_p^n.
  \label{eq:v-b-correspondence}
\end{equation}
The corresponding quadratic stabilizer state is
\begin{equation}
  \ket{\psi_{X,b}}
  :=p^{-n/2}\sum_{x\in\F_p^n}
    \omega^{\frac12x^TXx+b^Tx}\ket{x}.
  \label{eq:quadratic-stabilizer-state}
\end{equation}
A direct calculation gives the eigenvalue-character relation
\begin{equation}
  W(u,Xu)\ket{\psi_{X,b}}
  =\omega^{-b^Tu}\ket{\psi_{X,b}}
  =\omega^{v_{X,b}((u,Xu))}\ket{\psi_{X,b}}.
  \label{eq:eigenvalue-character}
\end{equation}
Accordingly, the abstract-character notation $(X,v)$ and the quadratic-state
notation $(X,b)$ describe the same secret data through
Eq.~\eqref{eq:v-b-correspondence}.  In the protocol sections we retain $v$
when emphasizing the abstract character and use $b$ when applying the
prediction theorem.

For fixed $X$, the family
\begin{equation}
  \mathcal B_X
  :=\{\ket{\psi_{X,b}}:b\in\F_p^n\}
  \label{eq:fixed-chart-basis}
\end{equation}
is an orthonormal basis of $\mathcal H_n$, since
\begin{equation}
  \braket{\psi_{X,b'}}{\psi_{X,b}}
  =p^{-n}\sum_{x\in\F_p^n}\omega^{(b-b')^Tx}
  =\delta_{b,b'}.
  \label{eq:fixed-chart-orthogonality}
\end{equation}
We refer to $X$ as the measurement chart and to $b$ as the outcome or
character label.

\subsection{Public and secret keys}
\label{subsec:keys-preliminaries}

The secret key is a pair
\begin{equation}
  \mathsf{sk}:=(X,v),
  \qquad
  X\in\Sym_n(\F_p),\quad v\in L_X^*,
  \label{eq:secret-key-preliminary}
\end{equation}
chosen according to the key-generation distribution specified in
Sec.~\ref{sec:protocol}.  Equivalently, using
Eq.~\eqref{eq:v-b-correspondence}, the secret key can be represented by
$(X,b)\in\Sym_n(\F_p)\times\F_p^n$.  One copy of the quantum public key is
\begin{equation}
  \mathsf{pk}:=\ket{\psi_{X,b}}.
  \label{eq:public-key-preliminary}
\end{equation}
The expression ``public key'' therefore refers to a physical copy of a
quantum state, not to an unrestricted classical description.  Copy counts and
the separation between adversarial copies and a verifier's independent test
copy will be specified in Secs.~\ref{sec:correctness-resources} and
\ref{sec:security-model}.

\subsection{Notation reserved for the protocol and security analysis}
\label{subsec:reserved-notation}

Table~\ref{tab:protocol-notation} collects the basic protocol and security
parameters.  In particular, $N$ is reserved for the adversarial copy count,
while $\ell$ denotes the dimension of the message subspace.  The symbols $d$
and $r$ are reserved here for later use; no reduction or security claim is made
before Sec.~\ref{sec:conditional-reduction}.

\begin{table*}[t]
\caption{Protocol parameters and basic notation.  All vector and matrix
operations are over $\F_p$.}
\label{tab:protocol-notation}
\centering
\renewcommand{\arraystretch}{1.15}
\begin{tabular}{@{}p{0.12\textwidth}p{0.24\textwidth}p{0.57\textwidth}@{}}
\hline\hline
Symbol & Type or definition & Meaning \\
\hline
$p$ & odd prime & Local qudit dimension and cardinality of the base field. \\
$\F_p$ & finite field & Field used for all vectors, matrices, ranks, and phases. \\
$n$ & positive integer & Number of qudits in one quantum public-key copy. \\
$\ell$ & $1\leq\ell<n$ & Dimension of the message-selected commuting family. \\
$r$ & $n-\ell$ & Residual dimension after conditioning on one valid signature. \\
$X$ & $\Sym_n(\F_p)$ & Secret symmetric matrix specifying the graph chart. \\
$b$ & $\F_p^n$ & Secret character parameter, with eigenvalue $\omega^{-b^Tu}$. \\
$\ket{\psi_{X,b}}$ & state on $n$ qudits & One physical copy of the quantum public key. \\
$Y$ & $\F_p^{(n-\ell)\times\ell}$ & Honestly signed message. \\
$Y'$ & $\F_p^{(n-\ell)\times\ell}$, $Y'\neq Y$ & Target message specified independently of the hidden key and exposed signature. \\
$A_Y$ & $(I_\ell;Y)$ & Full-column-rank embedding of the message subspace. \\
$Z_Y$ & $XA_Y$ & Chart component disclosed by the honest signature. \\
$h_Y$ & $-A_Y^Tb$ & Character component disclosed by the honest signature. \\
$\sigma_Y$ & $(Z_Y,h_Y)$ & Classical signature on $Y$. \\
$N$ & nonnegative integer & Number of public-key copies available to the adversary; the verifier uses a separate copy. \\
$d$ & $\rank(Y-Y')$ & Number of genuinely new residual target directions. \\
$P_{\mathrm{forge}}^{(N),*}(Y\to Y')$ & optimized probability & Optimal average acceptance probability for the specified target $Y'$. \\
\hline\hline
\end{tabular}
\end{table*}

\section{Authentication protocol}
\label{sec:protocol}

We now define the authentication protocol in the notation fixed in
Sec.~\ref{sec:preliminaries}.  The signer, Alice, keeps a classical description
of a quadratic stabilizer state, whereas each verifier receives a physical copy
of that state as a quantum public key.  A classical signature specifies a
message-dependent commuting Weyl family together with the eigenvalue character
that the public-key state must exhibit on that family.

\subsection{Message space and message-dependent observables}
\label{subsec:message-space}

Fix an integer $\ell$ satisfying $1\leq\ell<n$.  The message space is
\begin{equation}
  \mathcal M_{n,\ell}
  :=\F_p^{(n-\ell)\times\ell}.
  \label{eq:message-space}
\end{equation}
For $Y\in\mathcal M_{n,\ell}$, define the full-column-rank matrix
\begin{equation}
  A_Y
  :=\begin{pmatrix}I_\ell\\Y\end{pmatrix}
  \in\F_p^{n\times\ell}.
  \label{eq:AY-definition}
\end{equation}
Let $a_{Y,j}$ denote its $j$th column.  Given a secret symmetric matrix
$X\in\Sym_n(\F_p)$, set
\begin{equation}
  Z_Y:=XA_Y\in\F_p^{n\times\ell}
  \label{eq:ZY-definition}
\end{equation}
and
\begin{equation}
  G_Y
  :=\begin{pmatrix}A_Y\\Z_Y\end{pmatrix}
  =\begin{pmatrix}I_\ell\\Y\\Z_Y\end{pmatrix}
  \in\F_p^{2n\times\ell}.
  \label{eq:GY-definition}
\end{equation}
The $j$th column
\begin{equation}
  g_{Y,j}:=\binom{a_{Y,j}}{Xa_{Y,j}}
  \label{eq:gYj-definition}
\end{equation}
lies in the graph Lagrangian $L_X$.  Hence the Weyl operators
$W(g_{Y,1}),\ldots,W(g_{Y,\ell})$ commute.

Let the linear functional $v\in L_X^*$ correspond to $b\in\F_p^n$ through
Eq.~\eqref{eq:v-b-correspondence}.  The signature's character vector is
\begin{equation}
  h_Y
  :=\begin{pmatrix}
       v(g_{Y,1})\\ \vdots\\ v(g_{Y,\ell})
     \end{pmatrix}
  =-A_Y^Tb
  \in\F_p^\ell.
  \label{eq:hY-definition}
\end{equation}
By Eq.~\eqref{eq:eigenvalue-character},
\begin{equation}
  W(g_{Y,j})\ket{\psi_{X,b}}
  =\omega^{(h_Y)_j}\ket{\psi_{X,b}}
  \qquad (1\leq j\leq\ell).
  \label{eq:honest-message-eigenvalue}
\end{equation}

\subsection{Algorithms}
\label{subsec:protocol-algorithms}

The protocol consists of key generation, signing, and verification.

\begin{protocol}{Stabilizer-public-key authentication}
  \label{prot:authentication}
  \begin{algorithmic}[1]
    \Statex \textbf{Key generation $\mathsf{KeyGen}(1^n)$}
    \State Sample $X$ uniformly from $\Sym_n(\F_p)$.
    \State Sample $b$ uniformly from $\F_p^n$ and define
      $v((u,Xu)):=-b^Tu$.
    \State Set $\mathsf{sk}:=(X,v)$, equivalently $(X,b)$.
    \State Prepare the required finite collection of copies of
      $\mathsf{pk}:=\ket{\psi_{X,b}}$ and distribute them through an
      authenticated public-key distribution mechanism.
    \Statex
    \Statex \textbf{Signing $\mathsf{Sign}_{\mathsf{sk}}(Y)$}
    \Require $Y\in\mathcal M_{n,\ell}$.
    \State Form $A_Y$ by Eq.~\eqref{eq:AY-definition}.
    \State Compute $Z_Y:=XA_Y$ and $h_Y:=-A_Y^Tb$.
    \State Return the classical signature
      $\sigma_Y:=(Z_Y,h_Y)$.
    \Statex
    \Statex \textbf{Verification $\mathsf{Verify}_{\mathsf{pk}}(Y,\sigma)$}
    \Require One verifier-held copy of $\mathsf{pk}$, a message
      $Y\in\mathcal M_{n,\ell}$, and a purported signature $\sigma=(Z,h)$.
    \State Perform the classical well-formedness checks in
      Eqs.~\eqref{eq:well-formed-dimensions} and
      \eqref{eq:well-formed-isotropy}; reject if either check fails.
    \State Construct $g_j=(a_{Y,j},z_j)$ from the $j$th columns of
      $A_Y$ and $Z$.
    \State Measure the joint eigenspace projector
      $\Pi_{Y,Z,h}$ defined in Eq.~\eqref{eq:joint-accept-projector}.
    \State Accept on the $\Pi_{Y,Z,h}$ outcome and reject otherwise.
  \end{algorithmic}
\end{protocol}

The signature is entirely classical.  The public key is not a classical
matrix description of $(X,b)$; it is a finite physical supply of the state
$\ket{\psi_{X,b}}$.  The security model will therefore count the copies
available to an adversary explicitly.

\subsection{Classical well-formedness checks}
\label{subsec:well-formedness}

For a purported signature $(Z,h)$, the verifier first checks
\begin{equation}
  Z\in\F_p^{n\times\ell},
  \qquad
  h\in\F_p^\ell.
  \label{eq:well-formed-dimensions}
\end{equation}
Define
\begin{equation}
  G_{Y,Z}:=\begin{pmatrix}A_Y\\Z\end{pmatrix}.
  \label{eq:purported-G}
\end{equation}
The verifier also checks the isotropy condition
\begin{equation}
  Z^TA_Y-A_Y^TZ=0.
  \label{eq:well-formed-isotropy}
\end{equation}
Indeed, the left-hand side is the symplectic Gram matrix of the columns of
$G_{Y,Z}$.  Condition~\eqref{eq:well-formed-isotropy} is therefore equivalent
to pairwise commutativity of the proposed verification observables.  For an
honest signature,
\begin{equation}
  Z_Y^TA_Y-A_Y^TZ_Y
  =A_Y^T(X^T-X)A_Y=0.
  \label{eq:honest-isotropy}
\end{equation}
No candidate secret matrix $X'$ or character $v'$ is required in the
verification algorithm.  This point will also be reflected in the general
adversarial model of Sec.~\ref{sec:security-model}.

\begin{figure*}[t]
  \centering
  \includegraphics[width=0.85\textwidth]{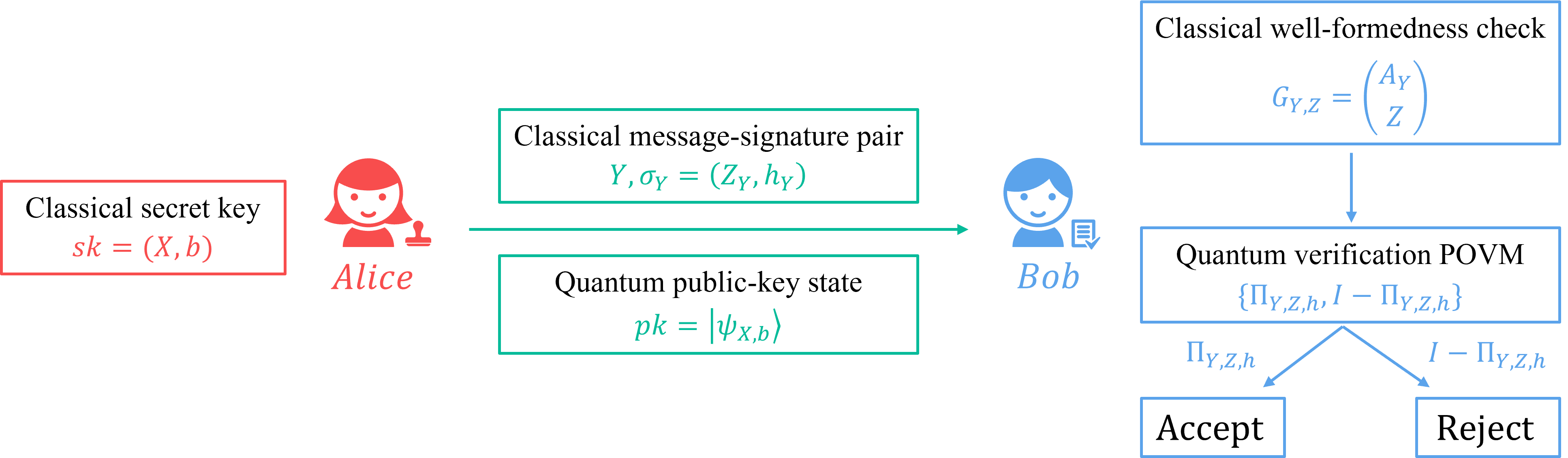}
  \caption{
  Overview of the stabilizer-public-key authentication protocol.
  Alice holds the classical secret description $(X,b)$ and sends a
  copy of the stabilizer state $\ket{\psi_{X,b}}$ as quantum
  public keys to the verifier.  For a message $Y$, Alice sends the classical pair
  $(Y,\sigma_Y)$ to the verifier.
  For a purported signature $(Z, h)$, the verifier checks its classical well-formedness first and uses an independently allocated public-key copy to make a joint verification POVM.
  }
  \label{fig:qap}
\end{figure*}

\section{Correctness and public-key resources}
\label{sec:correctness-resources}

We next give the verification measurement explicitly and separate correctness
from the later security analysis.  Correctness is an exact statement in the
ideal noiseless model.  It does not rely on the number of public-key copies
held by an adversary.

\subsection{The joint verification measurement}
\label{subsec:joint-verification-measurement}

For a phase-space vector $g\in\F_p^{2n}$ and $c\in\F_p$, define
\begin{equation}
  \Pi_{g,c}
  :=\frac1p\sum_{s\in\F_p}\omega^{-sc}W(g)^s.
  \label{eq:single-Weyl-projector}
\end{equation}
As shown in Appendix~\ref{app:verification-measurement},
$\Pi_{g,c}$ is the spectral projector of $W(g)$ associated with eigenvalue
$\omega^c$.  If $(Z,h)$ passes the well-formedness checks, the Weyl operators
indexed by the columns $g_1,\ldots,g_\ell$ of $G_{Y,Z}$ commute.  Their spectral
projectors therefore commute, and the accept projector is
\begin{equation}
  \Pi_{Y,Z,h}
  :=\prod_{j=1}^{\ell}\Pi_{g_j,h_j}.
  \label{eq:joint-accept-projector}
\end{equation}
The two-outcome verification POVM is
\begin{equation}
  \bigl\{\Pi_{Y,Z,h},\ I-\Pi_{Y,Z,h}\bigr\}.
  \label{eq:verification-POVM}
\end{equation}
Because the top block $A_Y$ has column rank $\ell$, the vectors
$g_1,\ldots,g_\ell$ are linearly independent.  Thus
$\Pi_{Y,Z,h}$ projects onto the simultaneous character sector specified by all
$\ell$ signed observables.

\subsection{Ideal correctness}
\label{subsec:ideal-correctness}

\begin{theorem}[Perfect correctness]
  \label{thm:perfect-correctness}
  Let $(\mathsf{sk},\mathsf{pk})$ be generated by
  $\mathsf{KeyGen}$, let $Y\in\mathcal M_{n,\ell}$, and let
  $\sigma_Y=\mathsf{Sign}_{\mathsf{sk}}(Y)$.  In the ideal noiseless model,
  verification using an independently allocated copy of $\mathsf{pk}$ accepts
  $(Y,\sigma_Y)$ with probability one:
  \begin{equation}
    \Pr\!\left[
      \mathsf{Verify}_{\mathsf{pk}}(Y,\sigma_Y)=\mathsf{accept}
    \right]=1.
    \label{eq:perfect-correctness}
  \end{equation}
\end{theorem}

\begin{proof}
  The honest signature has the correct dimensions and satisfies the isotropy
  check by Eq.~\eqref{eq:honest-isotropy}.  Its $j$th proposed observable is
  $W(g_{Y,j})$, and Eq.~\eqref{eq:honest-message-eigenvalue} gives
  \begin{equation}
    \Pi_{g_{Y,j},(h_Y)_j}\ket{\psi_{X,b}}
    =\ket{\psi_{X,b}}
    \qquad (1\leq j\leq\ell).
    \label{eq:correctness-single-projector}
  \end{equation}
  The projectors commute, so their product also fixes the public-key state:
  \begin{equation}
    \Pi_{Y,Z_Y,h_Y}\ket{\psi_{X,b}}
    =\ket{\psi_{X,b}}.
    \label{eq:correctness-joint-projector}
  \end{equation}
  Born's rule then gives acceptance probability one.
\end{proof}

\begin{remark}[Scope of correctness]
  \label{rem:correctness-scope}
  Theorem~\ref{thm:perfect-correctness} assumes exact state preparation,
  transmission, storage, and measurement.  A noise-tolerant acceptance test
  would require thresholds and a separate completeness--soundness analysis;
  no such claim is made in the present manuscript.
\end{remark}

\subsection{Finite-copy public-key resource model}
\label{subsec:public-key-resources}

A quantum public key is a physical state, so the signer prepares a finite
number of nominally identical copies.  We distinguish three roles:
\begin{enumerate}[label=(\roman*),leftmargin=*]
  \item \emph{Verifier copies}: an honest verification instance uses one copy
        allocated to that verifier.
  \item \emph{Adversarial copies}: the security experiment gives the adversary
        exactly $N$ copies on which an arbitrary collective quantum operation
        may be performed.
  \item \emph{Signer resources}: the classical secret key remains with Alice;
        preparation of additional public-key copies is outside a verification
        instance and must respect the copy budget of the deployment.
\end{enumerate}
The adversary's $N$ copies do not include the independent copy reserved for the
final verification test.  Hence the complete forgery experiment uses at least
$N+1$ public-key copies: $N$ at the adversary's input and one at the verifier.
Additional copies held by unrelated honest verifiers play no role unless the
security model explicitly grants the adversary access to them.  Consequently,
the parameter $N$ must be interpreted as an enforced upper bound on the total
number of copies available to the attacking coalition, not merely as the
number obtained through one designated interface.  If recipients can transfer,
resell, lend, or pool additional copies, those copies must be charged to the
same adversarial budget.  The theorem gives no security conclusion after the
actual coalition supply exceeds the stated $N$.

This provisioning model also separates two notions that coincide for an
ordinary classical public key but not here.  Verification is public in the
sense that no secret classical information is required by a holder of an
authenticated copy.  Verification is not unlimited: each verification instance
requires access to a physical copy, and the security parameter depends on how
many such copies exist and who can aggregate them.

Although an ideal projective verification of an honest eigenstate leaves that
state in the accepted eigenspace, the present security theorem will not claim
safe public-key reuse after verification.  The analysis instead models each
verification with an independently allocated copy.  This avoids conflating
ideal non-disturbance with composable reuse security and is consistent with the
separate key-reuse limitations studied in Sec.~\ref{sec:key-reuse}.

\subsection{Authenticated distribution assumption}
\label{subsec:distribution-assumption}

The protocol requires a mechanism that binds a distributed quantum state to
Alice's public identity.  The present work analyzes forgery after legitimate
copies of $\ket{\psi_{X,b}}$ have been distributed; it does not analyze
replacement of those copies during distribution.  Accordingly, public-key
distribution is assumed to be authenticated.  This assumption is logically
separate from secrecy: the state is public in the sense that a bounded number
of copies may be supplied beyond the signer, but an adversary may not replace a
verifier's designated test copy by an arbitrary state.

\section{Single-signature forgery model}
\label{sec:security-model}

We now define the adversarial task studied in this paper.  The purpose of the
model is to isolate the effect of a finite number of quantum-public-key copies
when one valid message--signature pair has already been exposed.  The model
allows an arbitrary collective quantum attack and does not restrict the
adversary to reconstructing a candidate secret key.

\subsection{Attack interface and information available to the adversary}
\label{subsec:attack-interface}

Fix protocol parameters $(p,n,\ell)$ and two distinct messages
\begin{equation}
  Y,Y'\in\mathcal M_{n,\ell},
  \qquad Y'\neq Y.
  \label{eq:distinct-messages-security}
\end{equation}
The message $Y$ is the honestly signed message and $Y'$ is the target message; both are externally fixed independently of the hidden key.
The key-generation variables $(X,b)$ are sampled uniformly as specified in
Protocol~\ref{prot:authentication}.  The adversary is given
\begin{equation}
  \left(Y,\sigma_Y,Y',
  \rho_{X,b}^{\otimes N}\right),
  \qquad
  \rho_{X,b}:=\ket{\psi_{X,b}}\!\bra{\psi_{X,b}},
  \label{eq:adversary-input}
\end{equation}
where
\begin{equation}
  \sigma_Y=(Z_Y,h_Y)
  \label{eq:valid-signature-security}
\end{equation}
is one honestly generated signature and $N$ is the number of public-key copies
allocated to the adversary.  A separate independent copy of $\rho_{X,b}$ is
reserved for the final verification test and is not contained in the $N$ input
copies.

The target $Y'$ is included explicitly in the attack interface.  The primary
quantity analyzed below is therefore a success probability for the specified target for
$Y\to Y'$.  This formulation keeps the later reduction transparent.  An
adaptive target-selection variant is defined separately in
Sec.~\ref{subsec:adaptive-targets}; no adaptive security claim will be inferred
from a bound for a separately specified target without an explicit argument.

\subsection{General collective adversaries}
\label{subsec:general-adversary}

For fixed classical input $(Y,\sigma_Y,Y')$, a general adversary is a quantum
instrument whose classical output is a purported signature
\begin{equation}
  \widehat\sigma=(\widehat Z,\widehat h),
  \qquad
  \widehat Z\in\F_p^{n\times\ell},\quad
  \widehat h\in\F_p^\ell.
  \label{eq:forged-signature-output}
\end{equation}
Equivalently, after absorbing all private ancillas, intermediate measurements,
feed-forward, and classical randomization into one measurement, the adversary
may be represented by a POVM
\begin{equation}
  \mathsf M^{Y,\sigma_Y,Y'}
  =\left\{M_{\widehat Z,\widehat h}^{Y,\sigma_Y,Y'}
  \right\}_{\widehat Z,\widehat h}
  \label{eq:adversary-POVM}
\end{equation}
on $\mathcal H_n^{\otimes N}$ satisfying
\begin{equation}
  M_{\widehat Z,\widehat h}^{Y,\sigma_Y,Y'}\geq0,
  \qquad
  \sum_{\widehat Z,\widehat h}
  M_{\widehat Z,\widehat h}^{Y,\sigma_Y,Y'}
  =I_{\mathcal H_n^{\otimes N}}.
  \label{eq:adversary-POVM-normalization}
\end{equation}
The dependence of the POVM on all exposed classical information is allowed.
No computational restriction is imposed.  In particular, the adversary may
perform entangled measurements across all $N$ copies and need not output, even
internally, a symmetric matrix $X'$ or a character $v'$ consistent with the
valid signature.

We include every matrix and vector of the stated dimensions as a possible
classical output.  If $(\widehat Z,\widehat h)$ fails the verifier's classical
well-formedness tests, its conditional acceptance probability is zero.

\subsection{Forgery experiment}
\label{subsec:forgery-experiment}

For an adversary $\mathcal A$, a signed message $Y$, and a specified target
message $Y'\neq Y$, define the experiment
\begin{equation}
  \mathsf{Forge}_{\mathcal A}^{(N)}(p,n,\ell;Y\to Y')
  \label{eq:forge-experiment-name}
\end{equation}
by the following steps.

\begin{enumerate}[label=\textbf{F\arabic*.},leftmargin=*]
  \item Sample $X$ uniformly from $\Sym_n(\F_p)$ and $b$ uniformly from
        $\F_p^n$.  Set
        $\mathsf{sk}=(X,b)$ and
        $\rho_{X,b}=\ket{\psi_{X,b}}\!\bra{\psi_{X,b}}$.
  \item Compute the valid signature
        $\sigma_Y=\mathsf{Sign}_{\mathsf{sk}}(Y)$.
  \item Give $(Y,\sigma_Y,Y')$ and $N$ copies
        $\rho_{X,b}^{\otimes N}$ to $\mathcal A$.
  \item Let $\widehat\sigma=(\widehat Z,\widehat h)$ be the adversary's
        classical output.
  \item Give an independent copy of $\rho_{X,b}$, together with
        $(Y',\widehat\sigma)$, to the verifier.
  \item The experiment outputs $1$ exactly when
        $\mathsf{Verify}_{\mathsf{pk}}(Y',\widehat\sigma)$ accepts.
\end{enumerate}

The average success probability for the specified target is
\begin{equation}
  P_{\mathrm{forge}}^{(N)}
  (\mathcal A;Y\to Y')
  :=\Pr\!\left[
    \mathsf{Forge}_{\mathcal A}^{(N)}
    (p,n,\ell;Y\to Y')=1
  \right],
  \label{eq:message-pair-forgery-probability}
\end{equation}
where the probability averages over uniform key generation, the adversary's
instrument, and the verifier's measurement outcome.  The optimized message-pair
probability is
\begin{equation}
  P_{\mathrm{forge}}^{(N),*}(Y\to Y')
  :=\sup_{\mathcal A}
  P_{\mathrm{forge}}^{(N)}(\mathcal A;Y\to Y').
  \label{eq:optimized-message-pair-forgery}
\end{equation}

For later use, the same quantity can be expanded directly in terms of the
adversarial POVM and the verifier's accept projector.  Let
\begin{equation}
  a_{X,b}^{Y'}(\widehat Z,\widehat h)
  :=\Tr\!\left[
     \Pi_{Y',\widehat Z,\widehat h}\rho_{X,b}
  \right]
  \label{eq:conditional-verification-payoff}
\end{equation}
when $(\widehat Z,\widehat h)$ is well formed, and set it to zero otherwise.
Then
\begin{align}
  &P_{\mathrm{forge}}^{(N)}(\mathcal A;Y\to Y')
  \notag\\
  &\ =\frac{1}{|\Sym_n(\F_p)|p^n}
  \sum_{X\in\Sym_n(\F_p)}\sum_{b\in\F_p^n}
  \notag\\[-2pt]
  &\qquad\times
  \sum_{\widehat Z,\widehat h}
  \Tr\!\left[
    M_{\widehat Z,\widehat h}^{Y,\sigma_Y,Y'}
    \rho_{X,b}^{\otimes N}
  \right]
  a_{X,b}^{Y'}(\widehat Z,\widehat h).
  \label{eq:forgery-probability-expanded}
\end{align}
The valid signature $\sigma_Y$ appearing in the POVM label is itself a
function of $(X,b)$.  Thus Eq.~\eqref{eq:forgery-probability-expanded} does not
treat the exposed classical signature as independent side information.  The
conditional reduction in Sec.~\ref{sec:conditional-reduction} must process
this dependence explicitly.

\subsection{Adaptive target selection}
\label{subsec:adaptive-targets}

A stronger variant permits the adversary to choose the target after observing
the valid pair.  Formally, the adversary first applies a classical randomized
map
\begin{equation}
  (Y,\sigma_Y)\longmapsto Y'\in
  \mathcal M_{n,\ell}\setminus\{Y\}
  \label{eq:adaptive-target-map}
\end{equation}
and then applies a target-dependent collective instrument to the $N$ quantum
copies.  We denote the corresponding optimized average success probability by
\begin{equation}
  P_{\mathrm{forge,ad}}^{(N),*}(Y).
  \label{eq:adaptive-forgery-probability}
\end{equation}
This definition is recorded to make the hierarchy of security notions clear.
The principal reduction will first be established for the quantity associated with the specified target
in Eq.~\eqref{eq:optimized-message-pair-forgery}.  An upper bound on
Eq.~\eqref{eq:adaptive-forgery-probability} will be stated only if the later
analysis controls the target selection uniformly; an average bound for a separately specified target alone is not declared to imply adaptive security.

\subsection{Average security versus pointwise security}
\label{subsec:average-versus-worst-case}

The probability in Eq.~\eqref{eq:message-pair-forgery-probability} is an average
over the uniformly generated secret key.  The corresponding pointwise
quantity for a fixed key $(X,b)$ is
\begin{equation}
  P_{\mathrm{forge}\mid X,b}^{(N),*}(Y\to Y').
  \label{eq:pointwise-forgery-probability}
\end{equation}
No bound on Eq.~\eqref{eq:pointwise-forgery-probability} follows merely from an
average bound without an additional symmetry or concentration argument.  The
independent partial-stabilizer-learning theorem used later is also formulated for a
uniform parameter prior.  Accordingly, the principal claim of the present
paper is initially an average-key security upper bound.  We do not call that
result worst-case unforgeability.

Likewise, a claim for one valid signature is not an EUF-CMA claim.  Multiple
valid signatures reveal additional classical linear information about the
secret chart and character; this separate limitation is analyzed in
Sec.~\ref{sec:key-reuse}.

\begin{figure*}[t]
  \centering
  \includegraphics[width=0.80\textwidth]{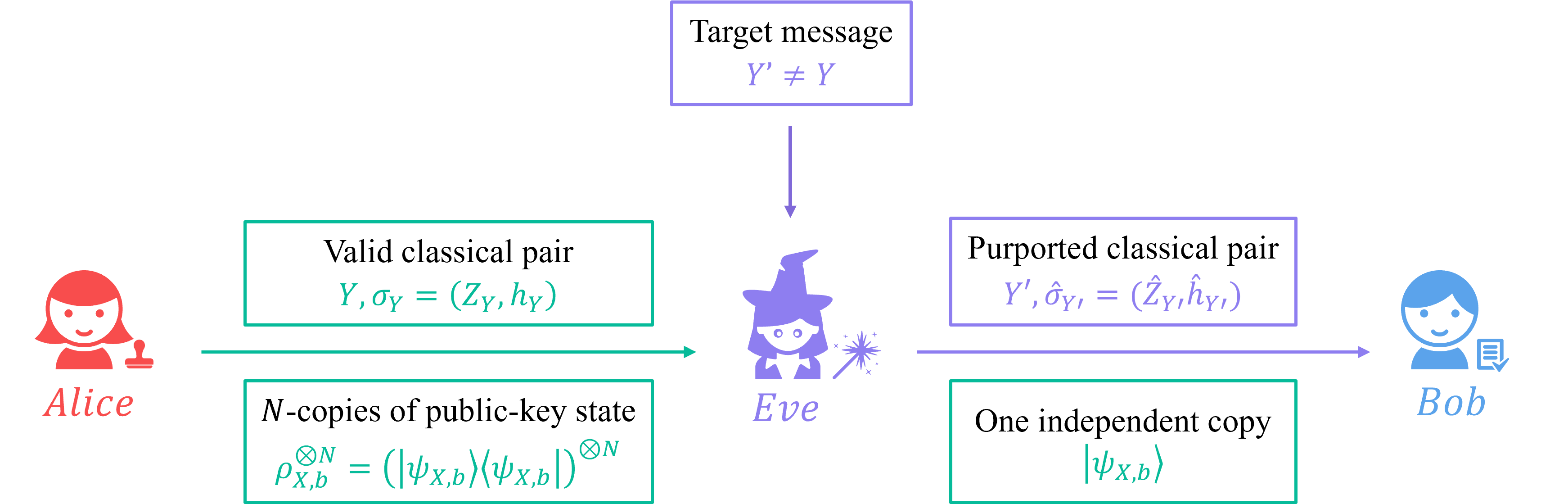}
  \caption{Single-signature forgery for the specified target interface.  Eve receives one valid
  classical pair $(Y,\sigma_Y)$, a distinct target message $Y'$, and $N$
  copies of the public-key state.  Eve may perform an arbitrary collective
  quantum operation and outputs a purported classical signature
  $\widehat\sigma=(\widehat Z,\widehat h)$.  A verifier then tests
  $(Y',\widehat\sigma)$ using one additional independent public-key copy  
  that is not included in Eve's $N$ copies.}
  \label{fig:adversarial-model}
\end{figure*}

\section{Reduction to partial quadratic-stabilizer prediction}
\label{sec:conditional-reduction}

We now upper-bound the single-signature forgery experiment by a residual
prediction problem.  The reduction has two logically distinct parts.  First, conditioning
on the exposed valid signature leaves an exactly standard quadratic
stabilizer ensemble on $r=n-\ell$ qudits.  Second, a target message at rank
distance $d$ asks only for the restriction of the residual chart and character
to a fixed $d$-dimensional subspace.  This second task is exactly the graded prediction-and-verification game for a specified query developed
in the independent learning paper; complete-label learning is its endpoint.

\subsection{Canonical coordinates for the signed message}
\label{subsec:canonical-signed-message}

Fix $Y\in\mathcal M_{n,\ell}$ and let
$A_Y=(I_\ell;Y)$.  Choose a matrix
$C_Y\in\F_p^{n\times r}$, where $r:=n-\ell$, such that
\begin{equation}
  R_Y:=[A_Y\ C_Y]\in\operatorname{GL}(n,\F_p).
  \label{eq:RY-definition}
\end{equation}
The complement $C_Y$ is fixed by an arbitrary deterministic convention.  
Different choices produce unitarily equivalent residual descriptions and do  
not affect any optimized game value.  
Set
\begin{equation}
  E_\ell:=\begin{pmatrix}I_\ell\\0\end{pmatrix}
  \in\F_p^{n\times\ell}.
  \label{eq:Eell-definition}
\end{equation}
Then $R_Y^{-1}A_Y=E_\ell$.  The phase-space transformation
\begin{equation}
  \mathsf S_Y
  :=\begin{pmatrix}R_Y^{-1}&0\\0&R_Y^T\end{pmatrix}
  \label{eq:SY-symplectic}
\end{equation}
is symplectic.  It sends the graph chart $L_X$ to the graph chart associated
with
\begin{equation}
  \widetilde X:=R_Y^TXR_Y
  \label{eq:Xtilde-definition}
\end{equation}
and sends the character vector to
\begin{equation}
  \widetilde b:=R_Y^Tb.
  \label{eq:btilde-definition}
\end{equation}

Using the decomposition
$\F_p^n=\F_p^\ell\oplus\F_p^r$, write
\begin{equation}
  \widetilde X=
  \begin{pmatrix}P&Q\\Q^T&S\end{pmatrix},
  \qquad
  \widetilde b=\binom{c}{e}.
  \label{eq:residual-block-decomposition}
\end{equation}
Here $P=P^T$, $S=S^T$,
$Q\in\F_p^{\ell\times r}$,
$c\in\F_p^\ell$, and $e\in\F_p^r$.

\begin{lemma}[Information exposed by one valid signature]
  \label{lem:signature-exposes-blocks}
  The classical data $(Y,Z_Y,h_Y)$ determine $(P,Q,c)$ through
  \begin{equation}
    R_Y^TZ_Y=\binom{P}{Q^T},
    \qquad
    h_Y=-c.
    \label{eq:signature-block-information}
  \end{equation}
  Under uniform key generation, conditioned on any admissible value of the
  exposed signature, $S$ and $e$ remain independent and uniform on
  $\Sym_r(\F_p)$ and $\F_p^r$, respectively.
\end{lemma}

\begin{proof}
  Since $A_Y=R_YE_\ell$,
  \begin{equation}
    R_Y^TZ_Y
    =R_Y^TXA_Y
    =\widetilde XE_\ell
    =\binom{P}{Q^T}.
    \label{eq:block-signature-proof}
  \end{equation}
  Likewise,
  $h_Y=-A_Y^Tb=-E_\ell^T\widetilde b=-c$.
  Congruence by $R_Y$ is a bijection of $\Sym_n(\F_p)$, and multiplication by
  $R_Y^T$ is a bijection of $\F_p^n$.  Therefore uniform independent $(X,b)$
  induce uniform independent $(P,Q,S,c,e)$.  Conditioning on $(P,Q,c)$ fixes
  the exposed coordinates and leaves $(S,e)$ uniform and independent.
\end{proof}

\subsection{Exact residual-state ensemble}
\label{subsec:exact-residual-ensemble}

Define the computational-basis permutation unitary
\begin{equation}
  U_{R_Y}\ket{x}:=\ket{R_Y^{-1}x},
  \qquad x\in\F_p^n.
  \label{eq:RY-basis-permutation}
\end{equation}
This unitary implements the symplectic map $\mathsf S_Y$ in
Eq.~\eqref{eq:SY-symplectic} exactly in the phase convention used here: after
writing the output coordinate as $y=R_Y^{-1}x$, the quadratic and linear
coefficients become $R_Y^TXR_Y$ and $R_Y^Tb$, with no additional affine term or
coordinate-dependent Clifford phase.  Define the diagonal unitary
\begin{align}
  D_{P,Q,c}\ket{x_1,x_2}
  :=&\ \omega^{-\frac12x_1^TPx_1-x_1^TQx_2-c^Tx_1}
  \ket{x_1,x_2},
  \label{eq:known-phase-removal}
\end{align}
where $x_1\in\F_p^\ell$ and $x_2\in\F_p^r$.
Both $U_{R_Y}$ and $D_{P,Q,c}$ are determined by the exposed classical data.

\begin{theorem}[Conditional residual-state reduction]
  \label{thm:conditional-residual-state}
  For every key consistent with the exposed valid signature,
  \begin{equation}
    D_{P,Q,c}U_{R_Y}\ket{\psi_{X,b}}
    =\ket{+}^{\otimes\ell}\otimes\ket{\psi_{S,e}},
    \label{eq:exact-residual-factorization}
  \end{equation}
  up to a global phase, where
  \begin{equation}
    \ket{\psi_{S,e}}
    =p^{-r/2}\sum_{z\in\F_p^r}
      \omega^{\frac12z^TSz+e^Tz}\ket{z}.
    \label{eq:residual-quadratic-state}
  \end{equation}
  Conditioned on the valid signature, $(S,e)$ is uniform on
  $\Sym_r(\F_p)\times\F_p^r$.
\end{theorem}

\begin{proof}
  By Eq.~\eqref{eq:RY-basis-permutation}, applying $U_{R_Y}$ replaces the
basis coordinate $x$ by $R_Yy$.  Hence, in the block coordinate
$y=(x_1,x_2)$, the computational-basis phase is
  \begin{equation}
    \frac12x_1^TPx_1+x_1^TQx_2+\frac12x_2^TSx_2+c^Tx_1+e^Tx_2.
    \label{eq:transformed-state-phase}
  \end{equation}
  The unitary $D_{P,Q,c}$ cancels exactly the terms involving the known
  coordinates $(P,Q,c)$.  The remaining amplitude factorizes into a uniform
  superposition over $x_1$ and the residual quadratic state over $x_2$.
  Uniformity of $(S,e)$ follows from
  Lemma~\ref{lem:signature-exposes-blocks}.
\end{proof}

Because the transformation in Theorem~\ref{thm:conditional-residual-state} is
known to the adversary, it can be applied independently to all $N$ copies.
The known factor $\ket{+}^{\otimes\ell N}$ can be discarded or supplied for
free.  Hence the attacker's quantum input is operationally equivalent to
$\ket{\psi_{S,e}}^{\otimes N}$ drawn from the standard uniform residual
ensemble.

\subsection{Target rank and partial residual data}
\label{subsec:target-rank-partial-data}

For the target message $Y'$, define
\begin{equation}
  B:=R_Y^{-1}A_{Y'}=\binom{B_1}{B_2},
  \label{eq:target-B-definition}
\end{equation}
where $B_1\in\F_p^{\ell\times\ell}$ and
$B_2\in\F_p^{r\times\ell}$.

\begin{lemma}[Target rank identity]
  \label{lem:target-rank-identity}
  Let $d:=\rank(Y-Y')$.  Then
  \begin{equation}
    \rank B_2=d.
    \label{eq:B2-rank-d}
  \end{equation}
  Equivalently,
  \begin{equation}
    \dim\bigl(\operatorname{col}(A_Y)
      \cap\operatorname{col}(A_{Y'})\bigr)=\ell-d.
    \label{eq:message-subspace-intersection}
  \end{equation}
\end{lemma}

\begin{proof}
  A vector belongs to both graph subspaces exactly when there exists
  $u\in\F_p^\ell$ such that
  \begin{equation}
    \binom{u}{Yu}=\binom{u}{Y'u},
    \label{eq:graph-intersection-vector}
  \end{equation}
  which is equivalent to $(Y-Y')u=0$.  Thus the intersection dimension is
  $\dim\ker(Y-Y')=\ell-d$.  Since $R_Y^{-1}$ sends $\operatorname{col}(A_Y)$ to $\F_p^\ell\oplus\{0\}$, the lower block $B_2$ is the matrix of the quotient projection of $\operatorname{col}(A_{Y'})$ onto $\F_p^n/\operatorname{col}(A_Y)$.  Hence $\ker B_2$ corresponds exactly to $\operatorname{col}(A_Y)\cap\operatorname{col}(A_{Y'})$.  Rank--nullity therefore gives $\rank B_2=\ell-(\ell-d)=d$.
\end{proof}

The honest transformed target chart action is
\begin{equation}
  \widetilde X B
  =\begin{pmatrix}P&Q\\Q^T&S\end{pmatrix}
   \binom{B_1}{B_2}.
  \label{eq:honest-target-action-transformed}
\end{equation}
Subtracting the known chart
\begin{equation}
  K:=\begin{pmatrix}P&Q\\Q^T&0\end{pmatrix}
  \label{eq:known-chart-K}
\end{equation}
leaves
\begin{equation}
  (\widetilde X-K)B=\binom{0}{SB_2}.
  \label{eq:unknown-target-chart-data}
\end{equation}
Similarly, after subtracting the known character contribution from $c$, the
unknown target outcomes are determined by $B_2^Te$.  Thus the genuinely new
target data are
\begin{equation}
  (SB_2,\ -B_2^Te),
  \label{eq:partial-target-data}
\end{equation}
up to known affine transformations and invertible changes of target basis.
\begin{figure*}[t]
\centering
\begin{tikzpicture}[
  font=\footnotesize,
  box/.style={draw,rounded corners,align=center,inner sep=4pt,minimum height=9mm},
  known/.style={box,fill=blue!8},
  hidden/.style={box,fill=orange!12},
  process/.style={box,fill=gray!8},
  comp/.style={box,fill=green!8},
  arr/.style={-{Latex[length=1.8mm]},thick}
]
\node[anchor=west,font=\bfseries] at (0,3.35) {(a) Exact conditional factorization};
\node[known,minimum width=40mm] (input) at (2.1,2.35)
  {exposed $(Y,Z_Y,h_Y)$\\and $\ket{\psi_{X,b}}$};
\node[process,minimum width=38mm] (unitary) at (8.0,2.35)
  {known unitary\\$D_{P,Q,c}U_{R_Y}$};
\node[hidden,minimum width=46mm] (factor) at (14.2,2.35)
  {$\ket{+}^{\otimes\ell}\otimes\ket{\psi_{S,e}}$\\$r=n-\ell$};
\draw[arr] (input)--(unitary);
\draw[arr] (unitary)--node[above,font=\scriptsize]{exact}(factor);
\node[known,minimum width=57mm] (blocks) at (5.0,0.95)
  {$\widetilde X=\left(\begin{smallmatrix}P&Q\\Q^T&S\end{smallmatrix}\right)$,
   $\widetilde b=\binom{c}{e}$\\exposed: $(P,Q,c)$};
\node[hidden,minimum width=58mm] (resid) at (12.1,0.95)
  {hidden residual: $(S,e)$\\uniform on $\Sym_r(\F_p)\times\F_p^r$};
\node[anchor=west,font=\bfseries] at (0,-0.15)
  {(b) Target query and one-directional security reduction};
\node[process,minimum width=16mm] (target) at (0.9,-1.15) {$Y'$};
\node[process,minimum width=18mm] (b2) at (3.0,-1.15) {$B_2$};
\node[known,minimum width=39mm] (rank) at (6.4,-1.15)
  {$\rank B_2=d=\rank(Y-Y')$};
\draw[arr] (target)--(b2); \draw[arr] (b2)--(rank);
\node[process,minimum width=28mm] (forge) at (1.45,-2.55)
  {well-formed\\forged output};
\node[process,minimum width=28mm] (retain) at (5.05,-2.55)
  {retain $d$ new\\directions};
\node[process,minimum width=28mm] (compress) at (8.65,-2.55)
  {compress\\known factor};
\node[process,minimum width=33mm] (complete) at (12.7,-2.55)
  {randomized isotropic\\completion};
\node[comp,minimum width=28mm] (legal) at (16.5,-2.55)
  {legal partial\\announcement};
\draw[arr] (forge)--(retain); \draw[arr] (retain)--(compress);
\draw[arr] (compress)--(complete); \draw[arr] (complete)--(legal);
\node[known,minimum width=55mm] (learn) at (3.5,-4.15)
  {\textbf{Finite-game correspondence}\\
   $P_{\mathrm{part}}^{(N),*}(r,d;p)
    =S_{\mathrm{ver}}^{(N),*}(r,d;p)$};
\node[comp,minimum width=82mm] (ineq) at (12.5,-4.15)
  {\textbf{Authentication-specific consequence}\\
   $P_{\mathrm{forge}}^{(N),*}(Y\!\to\!Y')
    \le P_{\mathrm{part}}^{(N),*}(r,d;p)$};
\draw[arr] (legal.south) -- ([xshift=40mm]ineq.north);
\end{tikzpicture}
\caption{Conditional residual-state reduction and its security interface.
Conditioning on one valid signature exposes $(P,Q,c)$, while the residual
parameters $(S,e)$ remain independent and uniform. A known unitary separates
each public-key copy into a known $\ket{+}$ factor and an $r$-qudit quadratic
stabilizer state. The specified target introduces exactly $d=\rank(Y-Y')$ new
residual directions. Compression and randomized isotropic completion yield
the one-directional bound from authentication forgery to the residual
partial-prediction game. The equality with the graded-verification game is the
separate finite-game correspondence to the companion symmetric-matrix graded-verification task.}
\label{fig:conditional-residual-reduction}
\end{figure*}
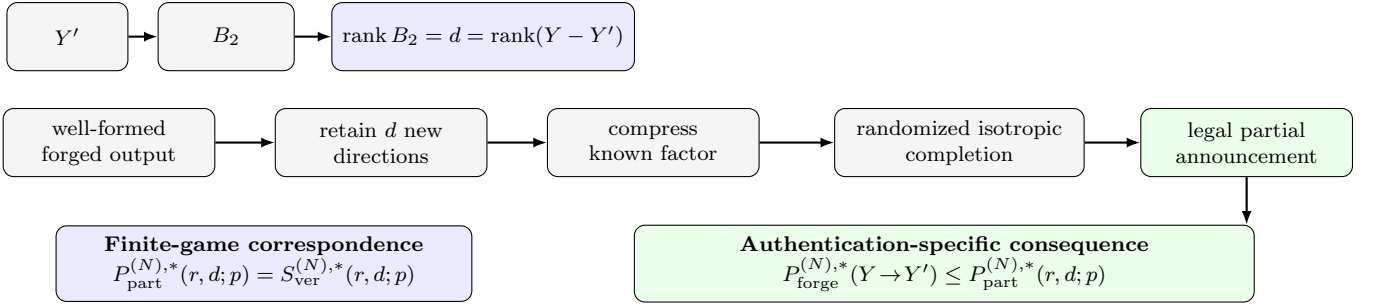

\subsection{Residual graded prediction-and-verification game}
\label{subsec:partial-prediction-game}
After conditioning on the exposed signature, the adversary has $N$ copies of
an unknown residual state $\rho_{S,e}$, while the specified target determines a
known $d$-dimensional query subspace.  The learner need not identify the full
residual label $(S,e)$.  It announces only a commuting family and character
for the queried directions, and an independent residual copy assigns the
operational score through a joint-character test.  We now define this residual
graded prediction-and-verification game in one place.

\begin{definition}[Residual graded prediction-and-verification game]
  \label{def:residual-graded-game}
  Fix an odd prime $p$, integers $0\le d\le r$ and $N\ge0$, and a full-column-
  rank matrix
  \begin{equation}
    J\in\F_p^{r\times d}.
    \label{eq:J-query-subspace}
  \end{equation}
  The query $J$ is fixed before the hidden residual label is sampled and is
  known to both the learner and the verifier.  Define the admissible chart
  space and announcement alphabet by
  \begin{align}
    \mathcal C_J
    &:=\{C\in\F_p^{r\times d}:J^TC=C^TJ\},
    \label{eq:game-admissible-chart-space}\\
    \Theta_J^{\rm part}
    &:=\mathcal C_J\times\F_p^d.
    \label{eq:game-announcement-alphabet}
  \end{align}
  The game proceeds as follows.
  \begin{enumerate}[label=\textbf{G\arabic*.},leftmargin=*]
    \item Nature samples $(S,e)$ uniformly from
          $\Sym_r(\F_p)\times\F_p^r$ and prepares
          \begin{equation}
            \rho_{S,e}:=\ket{\psi_{S,e}}\!\bra{\psi_{S,e}}.
            \label{eq:game-residual-state}
          \end{equation}
    \item The learner receives $N$ copies $\rho_{S,e}^{\otimes N}$ and performs
          an arbitrary collective POVM
          \begin{equation}
            \mathsf M=\{M_{C,\gamma}:(C,\gamma)\in
            \Theta_J^{\rm part}\},
            \label{eq:game-learner-POVM}
          \end{equation}
          where $M_{C,\gamma}\ge0$ and
          $\sum_{(C,\gamma)}M_{C,\gamma}=I$.  Its outcome is the partial
          announcement $(C,\gamma)$.
    \item The verifier receives one additional independent copy of
          $\rho_{S,e}$, not included in the learner's $N$ copies.  It applies
          the joint-character accept projector
          \begin{equation}
            \Pi_{J,C,\gamma}^{(d)}
            :=\frac1{p^d}\sum_{x\in\F_p^d}
              \omega^{-\gamma^Tx}W(Jx,Cx).
            \label{eq:game-accept-projector}
          \end{equation}
          The isotropy condition defining $\mathcal C_J$ makes
          $x\mapsto W(Jx,Cx)$ an ordinary representation of the additive
          group $\F_p^d$, so Eq.~\eqref{eq:game-accept-projector} is its joint
          spectral projector for the announced character
          $x\mapsto\omega^{\gamma^Tx}$.
  \end{enumerate}
  The conditional verification score is the Born acceptance probability
  \begin{equation}
    s_J\bigl((S,e),(C,\gamma)\bigr)
    :=\Tr\!\left[\Pi_{J,C,\gamma}^{(d)}\rho_{S,e}\right].
    \label{eq:game-conditional-score}
  \end{equation}
  For a learner POVM $\mathsf M$, its average score is
\begin{widetext}
\begin{equation}
S_J^{(N)}(\mathsf M;r,d,p)
:=\frac{1}{\lvert\Sym_r(\F_p)\rvert p^r}
\sum_{S\in\Sym_r(\F_p)}
\sum_{e\in\F_p^r}
\sum_{(C,\gamma)\in\Theta_J^{\rm part}}
\Tr\!\left[M_{C,\gamma}\rho_{S,e}^{\otimes N}\right]
s_J\bigl((S,e),(C,\gamma)\bigr).
\label{eq:game-average-graded-score}
\end{equation}
\end{widetext}
  The optimized score for a fixed query is
  \begin{equation}
    S_{\rm ver}^{(N),*}(J;p)
    :=\max_{\mathsf M}S_J^{(N)}(\mathsf M;r,d,p).
    \label{eq:game-specified-query-optimum}
  \end{equation}
  Private classical randomization after the measurement may be absorbed into
  a refinement of the POVM outcome.
\end{definition}

The correct partial label is
\begin{equation}
  q_J(S,e):=(SJ,-J^Te).
  \label{eq:game-true-partial-label}
\end{equation}
The first component gives the true residual chart action on the queried
subspace, and the second gives its eigenvalue character.

\begin{proposition}[Explicit rank formula for the residual graded score]
  \label{prop:partial-rank-payoff}
  Put
  \begin{equation}
    \Delta:=C-SJ,
    \qquad
    t:=\rank\Delta.
    \label{eq:partial-error-rank}
  \end{equation}
  Conditioned on $(S,e)$ and an admissible announcement $(C,\gamma)$, the
  verifier's score is
  \begin{equation}
    s_J\bigl((S,e),(C,\gamma)\bigr)
    =
    \begin{cases}
      p^{-t},&\gamma+J^Te\in\im(\Delta^T),\\
      0,&\gamma+J^Te\notin\im(\Delta^T).
    \end{cases}
    \label{eq:partial-rank-payoff}
  \end{equation}
\end{proposition}
\begin{proof}
  For $x\in\F_p^d$, the proposed observable is indexed by $(Jx,Cx)$.  Its
  expectation in $\ket{\psi_{S,e}}$ vanishes unless $(Jx,Cx)$ belongs to the
  true graph $L_S$, which is equivalent to $\Delta x=0$.  On
  $\ker\Delta$, the true character is $-e^TJx$.  Fourier inversion therefore
  shows that the announced joint outcome has nonzero probability exactly when
  \begin{equation}
    x^T(\gamma+J^Te)=0
    \quad\text{for every }x\in\ker\Delta.
    \label{eq:partial-compatibility-kernel}
  \end{equation}
  Since $(\ker\Delta)^\perp=\im(\Delta^T)$, this is the compatibility condition
  in Eq.~\eqref{eq:partial-rank-payoff}.  When it holds, the outcome
  distribution is uniform on an affine space of size $p^t$, giving score
  $p^{-t}$.
\end{proof}

\begin{corollary}[Score-one exactness]
  \label{cor:residual-score-one-exactness}
  For every true residual label and every admissible announcement,
  \begin{equation}
    s_J\bigl((S,e),(C,\gamma)\bigr)=1
    \quad\Longleftrightarrow\quad
    (C,\gamma)=(SJ,-J^Te).
    \label{eq:residual-score-one-exactness}
  \end{equation}
  Every false announcement has score either zero or at most $p^{-1}$.
\end{corollary}
\begin{proof}
  Score one in Proposition~\ref{prop:partial-rank-payoff} requires $t=0$, hence
  $C=SJ$, and then compatibility requires $\gamma+J^Te=0$.  Conversely, the
  true partial label has $t=0$ and satisfies compatibility.  If the chart is
  false, then $t\ge1$; if only the character is false, the score is zero.
\end{proof}

\begin{lemma}[Equivalence of fixed full-rank queries]
  \label{lem:specified-query-game-equivalence}
  For fixed $(r,d,N,p)$, all full-column-rank query matrices
  $J\in\F_p^{r\times d}$ have the same optimized graded score.  In particular,
  each is equivalent to the canonical query
  \begin{equation}
    J_d:=\begin{pmatrix}I_d\\0\end{pmatrix}\in\F_p^{r\times d}.
    \label{eq:game-canonical-query}
  \end{equation}
\end{lemma}
\begin{proof}
  Choose $R\in\operatorname{GL}(r,\F_p)$ with $J=RJ_d$.  The known basis
  permutation $U_R\ket{x}=\ket{R^{-1}x}$ maps
  $(S,e)$ bijectively to $(R^TSR,R^Te)$, preserves the uniform prior, maps an
  announcement $(C,\gamma)$ bijectively to $(R^TC,\gamma)$, and conjugates the
  accept projector for $J$ to that for $J_d$.  Applying $U_R$ to every learner
  copy and to the independent verifier copy therefore preserves each outcome
  probability and each conditional score.  The inverse coordinate change
  gives the reverse correspondence.
\end{proof}

By Lemma~\ref{lem:specified-query-game-equivalence}, we write the common optimum as
\begin{equation}
  S_{\mathrm{ver}}^{(N),*}(r,d;p)
  :=S_{\mathrm{ver}}^{(N),*}(J;p).
  \label{eq:partial-optimal-success}
\end{equation}
For compatibility with the authentication-specific reduction notation, we
also use the exact alias
\begin{equation}
  P_{\mathrm{part}}^{(N),*}(r,d;p)
  :=S_{\mathrm{ver}}^{(N),*}(r,d;p).
  \label{eq:partial-optimal-alias}
\end{equation}
The case $d=0$ is included only to close the mathematical definition under a
boundary case: it is an empty query with score one.  In the authentication
experiment, $Y'\ne Y$ implies $d=\rank(Y-Y')\ge1$.

The central objects are summarized in Table~\ref{tab:partial-prediction-notation}.
\begin{table*}[t]
\caption{Residual graded prediction-and-verification notation.  The true
partial label is $(SJ,-J^Te)$; consequently, the character error is
$\gamma+J^Te$.}
\label{tab:partial-prediction-notation}
\centering
\renewcommand{\arraystretch}{1.15}
\begin{tabular}{@{}p{0.14\textwidth}p{0.25\textwidth}p{0.54\textwidth}@{}}
\hline\hline
Symbol & Type or definition & Meaning \\
\hline
$R_Y$ & $[A_Y\ C_Y]\in\operatorname{GL}(n,\F_p)$ & Known coordinate change adapted to the exposed signed message. \\
$P,Q,c$ & exposed blocks & Chart and character coordinates determined by $(Y,Z_Y,h_Y)$. \\
$S$ & $\Sym_r(\F_p)$ & Hidden residual chart after conditioning on the valid signature. \\
$e$ & $\F_p^r$ & Hidden residual character parameter. \\
$\ket{\psi_{S,e}}$ & residual state & Unknown factor in $\ket{+}^{\otimes\ell}\otimes\ket{\psi_{S,e}}$. \\
$B_2$ & lower block of $R_Y^{-1}A_{Y'}$ & Quotient map carrying the genuinely new target directions; $\rank B_2=d$. \\
$J$ & $\F_p^{r\times d}$, full column rank & Public fixed residual query, equivalent to the canonical inclusion $J_d$. \\
$\mathcal C_J$ & $\{C:J^TC=C^TJ\}$ & Admissible announced chart actions defining commuting queried Weyl families. \\
$\Theta_J^{\rm part}$ & $\mathcal C_J\times\F_p^d$ & Learner announcement alphabet. \\
$(SJ,-J^Te)$ & true partial label & Residual chart action and character requested by the target. \\
$\Pi_{J,C,\gamma}^{(d)}$ & joint-character projector & Independent verifier's accept projector for announcement $(C,\gamma)$. \\
$\Delta$ & $C-SJ$ & Restricted chart error. \\
$t$ & $\rank\Delta$ & Rank shell of the chart error. \\
$\gamma+J^Te$ & character error & Compatibility requires membership in $\im(\Delta^T)$. \\
$s_J$ & Born acceptance probability & Equals $p^{-t}$ when compatible and $0$ otherwise. \\
$S_{\rm ver}^{(N),*}(r,d;p)$ & optimized average score & Optimal residual graded-verification score. \\
$P_{\rm part}^{(N),*}(r,d;p)$ & exact alias & Authentication-paper notation for the same optimized residual score. \\
\hline\hline
\end{tabular}
\end{table*}
\subsection{One-directional forgery-to-partial-prediction reduction}
\label{subsec:forgery-partial-reduction}

The reduction needed for security is one-directional.  A well-formed target
signature specifies an $\ell$-generator commuting Weyl family.  After an
invertible change of target-generator basis, the first $d$ generators carry
the new residual directions, whereas the final $\ell-d$ generators have zero
residual top component.  We delete the latter verification constraints.  This
marginal deletion can only increase the acceptance probability.

A forged retained family may couple the already exposed factor and the
residual factor even when neither projected family is commuting by itself.  We
therefore do not diagonalize the two factors separately.  Instead, we Fourier
expand the total commuting projector, compress it against the fixed exposed
state, and show that the compressed effect can be simulated exactly by a
randomized $d$-direction announcement in the partial prediction game.

\begin{lemma}[Retained block form and total isotropy]
  \label{lem:retained-total-isotropy}
  Fix the exposed valid signature and a well-formed purported target signature
  $\widehat\sigma=(\widehat Z,\widehat h)$.  Choose
  $V\in\operatorname{GL}(\ell,\F_p)$ such that
  \begin{equation}
    B_2V=\bigl[J\ \ 0\bigr],
    \qquad J\in\F_p^{r\times d},\quad\rank J=d,
    \label{eq:retained-B2-normal-form}
  \end{equation}
  and retain the first $d$ transformed generators.  After applying the known
  conjugation
  $\mathcal U_{Y,\sigma_Y}:=D_{P,Q,c}U_{R_Y}$ and absorbing its known linear
  phase into the transformed announced character
  $\eta\in\F_p^d$, the retained generators have labels
  \begin{equation}
    \left(
      \binom{Lx}{Jx},
      \binom{Fx}{Cx}
    \right),
    \qquad x\in\F_p^d,
    \label{eq:retained-total-label}
  \end{equation}
  for matrices
  \begin{equation}
    L,F\in\F_p^{\ell\times d},
    \qquad C\in\F_p^{r\times d}.
    \label{eq:retained-block-dimensions}
  \end{equation}
  To make the blocks and phase sign explicit, put
  \begin{equation}
    H:=V\begin{pmatrix}I_d\\0\end{pmatrix}
      \in\F_p^{\ell\times d},
    \qquad
    BH=\binom{L}{J},
    \label{eq:retained-explicit-top-blocks}
  \end{equation}
  and
  \begin{equation}
    R_Y^T\widehat ZH
      =\binom{\widehat F}{\widehat C}.
    \label{eq:retained-bottom-before-known-shear}
  \end{equation}
  Define
  \begin{equation}
    K_0:=\begin{pmatrix}P&Q\\Q^T&0\end{pmatrix},
    \qquad c_0:=\binom{c}{0}.
    \label{eq:retained-known-K0-c0}
  \end{equation}
  Direct conjugation by $D_{P,Q,c}$ gives
  \begin{equation}
    F=\widehat F-PL-QJ,
    \qquad
    C=\widehat C-Q^TL.
    \label{eq:retained-explicit-F-C}
  \end{equation}
  More precisely, for the $j$th retained top label
  $u_j=(L_j^T,J_j^T)^T$ and its pre-shear bottom label $z_j$,
  \begin{align}
    &D_{P,Q,c}W(u_j,z_j)D_{P,Q,c}^{\dagger}
      \notag\\[-2pt]
    &\qquad
      =\omega^{\kappa_j}W(u_j,z_j-K_0u_j),
    \qquad
      \kappa_j=-c^TL_j.
    \label{eq:retained-conjugation-phase}
  \end{align}
  Thus, if $\widehat h_d:=H^T\widehat h$ and
  $\kappa=(\kappa_1,\ldots,\kappa_d)^T$, the character after conjugation is
  \begin{equation}
    \eta=\widehat h_d-\kappa.
    \label{eq:retained-explicit-eta}
  \end{equation}
  The minus sign is forced because
  $\omega^{\kappa_j}W(u_j,z_j-K_0u_j)$ must retain the originally announced
  eigenvalue $\omega^{(\widehat h_d)_j}$.  The phase correction is linear in
  the retained top label.  For a Fourier coefficient $x\in\F_p^d$, the
  accumulated scalar phase is $\kappa^Tx$.
  These matrices and $\eta$ are deterministic functions of the exposed data,
  the target, and $\widehat\sigma$, and are independent of the hidden
  residual parameter $(S,e)$.  Their total symplectic Gram matrix vanishes:
  \begin{equation}
    F^TL-L^TF+C^TJ-J^TC=0.
    \label{eq:retained-total-isotropy}
  \end{equation}
\end{lemma}

\begin{proof}
  The generator-basis change induced by $V$ preserves the joint verification
  subspace after the character relabeling $\widehat h\mapsto V^T\widehat h$.
  The basis permutation $U_{R_Y}$ sends the retained top and bottom blocks to
  $BH$ and $R_Y^T\widehat ZH$ without an additional phase.  To verify the
  diagonal-conjugation phase, write
  $D_{P,Q,c}\ket{x}=\omega^{-\frac12x^TK_0x-c_0^Tx}\ket{x}$.
  Acting on a basis vector gives
  \begin{align*}
   & D_{P,Q,c}W(u,z)D_{P,Q,c}^{\dagger}\ket{x}\notag\\
    =&\omega^{z^Tx+\frac12u^Tz-x^TK_0u-\frac12u^TK_0u-c_0^Tu}
      \ket{x+u}.
  \end{align*}
  Comparing this expression with
  $W(u,z-K_0u)\ket{x}$ shows that the quadratic contribution is already
  contained in the shifted Weyl operator and that the remaining scalar is
  only $\omega^{-c_0^Tu}$.  Since $c_0^Tu_j=c^TL_j$, this proves
  $\kappa_j=-c^TL_j$ and hence
  Eqs.~\eqref{eq:retained-explicit-F-C}--\eqref{eq:retained-explicit-eta}.
  Splitting the transformed $n=\ell+r$ coordinates gives
  Eq.~\eqref{eq:retained-total-label}; the residual top block is $J$ by
  Eq.~\eqref{eq:retained-B2-normal-form}.  Conjugation and restriction to the
  retained generators preserve commutation of the total family.  Evaluating
  its symplectic Gram matrix in the direct-sum coordinates gives exactly
  Eq.~\eqref{eq:retained-total-isotropy}.  We do not infer that either summand
  in that equation vanishes separately.
\end{proof}

Let $\Pi_{\mathrm{full}}$ be the original accept projector after the target
basis change, and let $\Pi_{\mathrm{ret},\eta}$ retain only the first $d$
spectral constraints.  Since all original constraints commute,
\begin{equation}
  0\leq\Pi_{\mathrm{full}}\leq\Pi_{\mathrm{ret},\eta}.
  \label{eq:retained-projector-domination}
\end{equation}

\begin{lemma}[Fourier form and compression against the exposed factor]
  \label{lem:known-factor-compression}
  With the notation of Lemma~\ref{lem:retained-total-isotropy},
  \begin{equation}
    \Pi_{\mathrm{ret},\eta}
    =\frac1{p^d}\sum_{x\in\F_p^d}
      \omega^{-\eta^Tx}
      W_A(Lx,Fx)\otimes W_B(Jx,Cx).
    \label{eq:retained-Fourier-projector}
  \end{equation}
  Put
  \begin{equation}
    \tau_A:=\bigl(\ket{+}\!\bra{+}\bigr)^{\otimes\ell}
    \label{eq:known-factor-state}
  \end{equation}
  and define the effect induced on the residual verification system by
  \begin{equation}
    E_B(\widehat\sigma)
    :=\Tr_A\!\left[(\tau_A\otimes I_B)
      \Pi_{\mathrm{ret},\eta}\right].
    \label{eq:compressed-residual-effect}
  \end{equation}
  Then
  \begin{equation}
    E_B(\widehat\sigma)
    =\frac1{p^d}\sum_{x\in\ker F}
      \omega^{-\eta^Tx}W_B(Jx,Cx).
    \label{eq:compressed-kernel-Fourier-sum}
  \end{equation}
\end{lemma}

\begin{proof}
  The total retained labels are isotropic, so their Weyl operators form an
  ordinary representation of the additive group $\F_p^d$.  Fourier inversion
  gives Eq.~\eqref{eq:retained-Fourier-projector}; no commutativity of either
  factor family is used.  For every $u,z\in\F_p^\ell$,
  \begin{equation}
    \bra{+}^{\otimes\ell}W_A(u,z)\ket{+}^{\otimes\ell}
    =\begin{cases}1,&z=0,\\0,&z\ne0.\end{cases}
    \label{eq:plus-Weyl-expectation}
  \end{equation}
  Indeed, $W_A(u,0)=\mathsf X(u)$ fixes the uniform superposition, whereas a
  nonzero phase component has zero expectation by character orthogonality.
  Applying Eq.~\eqref{eq:plus-Weyl-expectation} termwise to
  Eq.~\eqref{eq:retained-Fourier-projector} leaves precisely those
  coefficients $x$ for which $Fx=0$, proving
  Eq.~\eqref{eq:compressed-kernel-Fourier-sum}.
\end{proof}

\begin{lemma}[Residual isotropy on the surviving kernel]
  \label{lem:kernel-residual-isotropy}
  Let
  \begin{equation}
    K:=\ker F,\qquad k:=\dim K,
    \label{eq:kernel-K-k}
  \end{equation}
  and let $T\in\F_p^{d\times k}$ have full column rank and image $K$.  Define
  \begin{equation}
    J_K:=JT,\qquad C_K:=CT,\qquad\gamma_K:=T^T\eta.
    \label{eq:kernel-restricted-labels}
  \end{equation}
  Then $\rank J_K=k$ and
  \begin{equation}
    J_K^TC_K=C_K^TJ_K.
    \label{eq:kernel-residual-isotropy}
  \end{equation}
  Moreover,
  \begin{equation}
    E_B(\widehat\sigma)
    =p^{-(d-k)}\Pi^{(k)}_{J_K,C_K,\gamma_K},
    \label{eq:kernel-projector-representation}
  \end{equation}
  where the projector on the right is the joint character projector of the
  $k$ commuting residual generators.
\end{lemma}

\begin{proof}
  Multiply Eq.~\eqref{eq:retained-total-isotropy} on the left by $T^T$ and on
  the right by $T$.  Since $FT=0$, the known-factor terms vanish and we obtain
  \begin{equation}
    (CT)^T(JT)-(JT)^T(CT)=0,
    \label{eq:kernel-isotropy-calculation}
  \end{equation}
  which is Eq.~\eqref{eq:kernel-residual-isotropy}.  Because $J$ and $T$ both
  have full column rank, so does $JT$.  Substituting the bijective
  parametrization $x=Ty$ of $K$ into
  Eq.~\eqref{eq:compressed-kernel-Fourier-sum} gives
  \begin{align}
    E_B(\widehat\sigma)
    &=\frac1{p^d}\sum_{y\in\F_p^k}
      \omega^{-\gamma_K^Ty}W_B(J_Ky,C_Ky)\notag\\
    &=p^{-(d-k)}\Pi^{(k)}_{J_K,C_K,\gamma_K},
    \label{eq:kernel-projector-proof}
  \end{align}
  proving Eq.~\eqref{eq:kernel-projector-representation}.  Replacing $T$ by
  $TR$ with $R\in\operatorname{GL}(k,\F_p)$ only changes the generator basis
  and the character by the corresponding bijective relabeling, so the
  projector is independent of the chosen basis of $K$.
\end{proof}

\begin{lemma}[Isotropic completion inside the fixed query space]
  \label{lem:isotropic-query-completion}
  Let $J\in\F_p^{r\times d}$ have full column rank.  Suppose
  $T\in\F_p^{d\times k}$ has full column rank and
  $C_K\in\F_p^{r\times k}$ satisfies
  \begin{equation}
    (JT)^TC_K=C_K^T(JT).
    \label{eq:completion-input-isotropy}
  \end{equation}
  Then there exists $C_{\mathrm{ext}}\in\F_p^{r\times d}$ such that
  \begin{equation}
    C_{\mathrm{ext}}T=C_K,
    \qquad
    J^TC_{\mathrm{ext}}=C_{\mathrm{ext}}^TJ.
    \label{eq:isotropic-completion-properties}
  \end{equation}
  The completion may be selected by a fixed deterministic linear-algebraic
  rule from $(J,T,C_K)$.
\end{lemma}

\begin{proof}
  Choose $R$ such that
  \begin{equation}
    V_c:=\begin{bmatrix}T&R\end{bmatrix}
      \in\operatorname{GL}(d,\F_p).
    \label{eq:completion-generator-basis}
  \end{equation}
  Since $JV_c$ has rank $d$, choose $U\in\operatorname{GL}(r,\F_p)$ such that
  \begin{equation}
    UJV_c=J_d:=\binom{I_d}{0}.
    \label{eq:completion-residual-normalization}
  \end{equation}
  The corresponding residual symplectic coordinate transformation is
  \begin{equation}
    (u,z)\longmapsto(Uu,U^{-T}z).
    \label{eq:completion-symplectic-coordinate-change}
  \end{equation}
  Hence the prescribed kernel chart in canonical coordinates is
  \begin{equation}
    C_K':=U^{-T}C_K.
    \label{eq:completion-prescribed-chart}
  \end{equation}
  In these coordinates an admissible completed chart has the form
  \begin{equation}
    C_{\mathrm{ext}}'=\binom{H}{R_0},
    \qquad H=H^T.
    \label{eq:completion-canonical-chart}
  \end{equation}
  Its prescribed first $k$ columns are $C_K'$.  Write their upper block as
  $(H_{11}^T,H_{21}^T)^T$, where
  $H_{11}\in\F_p^{k\times k}$.  The input isotropy condition is equivalent to
  $H_{11}=H_{11}^T$.  We may therefore choose
  \begin{equation}
    H=\begin{pmatrix}H_{11}&H_{21}^T\\H_{21}&0\end{pmatrix}
    \label{eq:canonical-symmetric-completion}
  \end{equation}
  and set all unspecified columns of $R_0$ to zero.  This produces a canonical
  $C_{\mathrm{ext}}'$ satisfying
  \begin{equation}
    C_{\mathrm{ext}}'\binom{I_k}{0}=C_K',
    \qquad
    J_d^TC_{\mathrm{ext}}'=(C_{\mathrm{ext}}')^TJ_d.
    \label{eq:canonical-completion-properties}
  \end{equation}
  The exact back-transport is
  \begin{equation}
    C_{\mathrm{ext}}
      :=U^TC_{\mathrm{ext}}'V_c^{-1}.
    \label{eq:completion-back-transport}
  \end{equation}
  Since $V_c^{-1}T=(I_k,0)^T$, one has
  $C_{\mathrm{ext}}T=C_K$.  Moreover,
  \begin{equation}
    V_c^T(J^TC_{\mathrm{ext}})V_c
      =(UJV_c)^TC_{\mathrm{ext}}'
      =J_d^TC_{\mathrm{ext}}',
    \label{eq:completion-isotropy-transport}
  \end{equation}
  which is symmetric.  Invertibility of $V_c$ proves
  $J^TC_{\mathrm{ext}}=C_{\mathrm{ext}}^TJ$.  Fixed row-reduction and basis-
  completion conventions make the construction deterministic.
\end{proof}

\begin{lemma}[Random-character completion]
  \label{lem:random-character-completion}
  Let $V=[T\ R]\in\operatorname{GL}(d,\F_p)$ and let
  $C_{\mathrm{ext}}$ be a completion from
  Lemma~\ref{lem:isotropic-query-completion}.  For
  $\zeta\in\F_p^{d-k}$ define
  \begin{equation}
    \gamma_{\mathrm{ext}}(\zeta)
    :=V^{-T}\binom{\gamma_K}{\zeta}.
    \label{eq:extended-random-character}
  \end{equation}
  Then
  \begin{equation}
    \frac1{p^{d-k}}\sum_{\zeta\in\F_p^{d-k}}
    \Pi^{(d)}_{J,C_{\mathrm{ext}},
      \gamma_{\mathrm{ext}}(\zeta)}
    =p^{-(d-k)}\Pi^{(k)}_{J_K,C_K,\gamma_K}.
    \label{eq:random-character-projector-identity}
  \end{equation}
  Consequently,
  \begin{equation}
    E_B(\widehat\sigma)
    =\mathbb E_{\zeta}\!
      \left[\Pi^{(d)}_{J,C_{\mathrm{ext}},
        \gamma_{\mathrm{ext}}(\zeta)}\right].
    \label{eq:compressed-effect-random-simulation}
  \end{equation}
\end{lemma}

\begin{proof}
  Work in the generator basis $V=[T\ R]$.  The completed $d$-generator family
  is commuting, and its first $k$ generators are exactly
  $(J_K;C_K)$.  Expanding each joint projector in characters and averaging
  over $\zeta$ gives
  \begin{equation}
    \frac1{p^{d-k}}\sum_{\zeta\in\F_p^{d-k}}
      \omega^{-\zeta^Tz}
    =\begin{cases}1,&z=0,\\0,&z\ne0.\end{cases}
    \label{eq:random-character-orthogonality}
  \end{equation}
  Thus all Fourier terms with a nonzero coefficient on an added generator
  vanish, while the terms supported on the first $k$ generators remain.  Since
  the full projector has coefficient $p^{-d}=p^{-(d-k)}p^{-k}$, the surviving
  sum is $p^{-(d-k)}$ times the normalized $k$-generator projector.  Hence the
  resulting normalization is $p^{-(d-k)}$, proving
  Eq.~\eqref{eq:random-character-projector-identity}.  Combining it with
  Eq.~\eqref{eq:kernel-projector-representation} proves
  Eq.~\eqref{eq:compressed-effect-random-simulation}.
\end{proof}
\begin{remark}[Boundary cases of the kernel reduction]
  \label{rem:kernel-boundary-cases}
  If $k=d$, then $\ker F=\F_p^d$ and hence $F=0$.  No generator or character
  is added, the randomization set is a singleton, and
  Eq.~\eqref{eq:random-character-projector-identity} reduces to the identity between
  the original $d$-generator projector and itself.  If $k=0$, then
  $\ker F=\{0\}$ and
  \begin{equation}
    E_B(\widehat\sigma)=p^{-d}I_B.
    \label{eq:k-zero-compressed-effect}
  \end{equation}
  Any fixed isotropic $d$-generator completion resolves the identity over all
  $p^d$ characters, so uniform randomization of all characters gives
  $p^{-d}\sum_\gamma\Pi_\gamma=p^{-d}I_B$, exactly reproducing the compressed
  effect.
\end{remark}

\begin{theorem}[Forgery-to-partial-prediction security reduction]
  \label{thm:forgery-partial-reduction}
  Fix distinct messages $Y,Y'\in\mathcal M_{n,\ell}$ and let
  \begin{equation}
    r=n-\ell,\qquad d=\rank(Y-Y').
    \label{eq:r-d-reduction-parameters}
  \end{equation}
  Then, for every $N\geq0$,
  \begin{equation}
    P_{\mathrm{forge}}^{(N),*}(Y\to Y')
    \leq P_{\mathrm{part}}^{(N),*}(r,d;p).
    \label{eq:forgery-partial-upper-bound}
  \end{equation}
  Only this upper-bound direction is asserted; no equality of the two
  operational optima is claimed.
\end{theorem}

\begin{proof}
  Fix an authentication adversary and condition on the exposed valid
  signature.  Apply the known copywise transformation from
  Theorem~\ref{thm:conditional-residual-state}.  Its input state becomes
  \begin{equation}
    \tau_A^{\otimes N}\otimes
    \rho_{S,e}^{\otimes N},
    \qquad
    \rho_{S,e}:=\ket{\psi_{S,e}}\!\bra{\psi_{S,e}},
    \label{eq:transformed-adversary-input}
  \end{equation}
  with uniform $(S,e)$.  If $\{M_q\}_q$ is the adversary's transformed POVM,
  define a POVM on the residual copies by
  \begin{equation}
    \widetilde M_q
    :=\Tr_{A^{\otimes N}}\!\left[
      (\tau_A^{\otimes N}\otimes I)M_q\right].
    \label{eq:effective-residual-POVM}
  \end{equation}
  Positivity and normalization follow from those of $\{M_q\}_q$.  Explicitly,
  $\sum_q\widetilde M_q=\Tr_{A^{\otimes N}}[(\tau_A^{\otimes N}\otimes I)I]=I$ because $\Tr\tau_A^{\otimes N}=1$.  The
  outcome distribution is unchanged for every $(S,e)$.

  For every well-formed outcome $q=(\widehat Z,\widehat h)$, apply the target
  basis change of Lemma~\ref{lem:retained-total-isotropy}, retain the first
  $d$ generators, and compute $(L,F,J,C,\eta)$.  Let $T$ be the basis of
  $\ker F$ selected by a fixed row-reduction convention.  Construct
  $C_K=CT$, an isotropic completion $C_{\mathrm{ext}}$, and then sample
  $\zeta\in\F_p^{d-k}$ uniformly.  Output the legal $d$-direction partial
  announcement
  \begin{equation}
    \bigl(C_{\mathrm{ext}},
      \gamma_{\mathrm{ext}}(\zeta)\bigr).
    \label{eq:reduction-partial-output}
  \end{equation}
  This randomized classical post-processing depends only on the exposed data,
  the target, the adversarial outcome, and private randomness; it is
  independent of $(S,e)$.  Malformed outcomes are mapped to an arbitrary fixed
  legal partial announcement.  Their authentication payoff is zero, whereas every partial-game payoff is nonnegative, so this replacement preserves the pointwise upper bound.

  By Eq.~\eqref{eq:retained-projector-domination}, authentication acceptance is
  at most the expectation of the retained projector.  By
  Eq.~\eqref{eq:compressed-effect-random-simulation}, after taking the known
  factor in state $\tau_A$, that expectation is exactly the average
  verification payoff of the randomized $d$-direction partial announcement
  in Eq.~\eqref{eq:reduction-partial-output}.  Hence, pointwise in $(S,e)$ and
  every adversarial outcome,
  \begin{equation}
    \Pr[\mathrm{accept}_{\mathrm{auth}}\mid S,e,q]
    \leq
    \mathbb E_{\zeta}
    \Pr[\mathrm{accept}_{\mathrm{part}}\mid S,e,q,\zeta].
    \label{eq:pointwise-auth-part-comparison}
  \end{equation}
  Averaging over the residual prior and all adversarial randomness produces an
  $N$-copy partial-game strategy whose success probability upper-bounds that
  of the fixed authentication adversary.  Taking the supremum over
  authentication adversaries proves
  Eq.~\eqref{eq:forgery-partial-upper-bound}.
\end{proof}

\begin{remark}[Graded prediction and verification is the relevant interface]
  \label{rem:full-label-not-enough}
  For $d<r$, authentication is governed by a genuine partial-label decision
  problem rather than by complete-label identification.  The independent learning paper treats exactly this partial game for a specified query and
  provides the finite-size and asymptotic results used here, so no
  reduction to rank-one complete-label projectors is needed.
\end{remark}

\section{Finite-copy security consequences and the copy-rate threshold}
\label{sec:finite-copy-security}

The preceding section completes the authentication-specific reduction.  The
remaining quantum task is not merely analogous to the independent learning paper's
partial-learning problem: after an explicit relabeling, it is exactly the same
finite symmetric-matrix stabilizer graded-verification game, formulated as the symmetric-matrix graded-verification task in the companion paper.  We first record this operational
identity, then import the finite rank-truncated comparison of the independent learning paper and its
subcritical converse.  For the opposite direction, we give an explicit
adversary that identifies the complete residual stabilizer label and
reconstructs the honest target signature.  Together these arguments determine
the forgery for the specified target copy rate away from the critical line
~\cite{HayashiQuadraticStabilizerPrediction}.

\subsection{Identification with the companion symmetric-matrix graded-verification task}
\label{subsec:partial-dimension-shell-upper-bound}

A known invertible residual coordinate change sends every fixed full-column-
rank query to the canonical inclusion
\begin{equation}
  J_d:=\begin{pmatrix}I_d\\0\end{pmatrix}\in\F_p^{r\times d}.
  \label{eq:security-Jd}
\end{equation}
For later use, put
\begin{equation}
 \mathcal C_{r,d}
 :=\{C\in\F_p^{r\times d}:J_d^TC=C^TJ_d\}.
 \label{eq:partial-rank-count}
\end{equation}
The corresponding partial-label space has cardinality
\begin{equation}
 M_{r,d,p}
 :=p^{(r+1)d-d(d-1)/2}.
 \label{eq:partial-label-cardinality-security}
\end{equation}

In the companion paper, an announced ordered Pauli reference frame specifies
the requested directions intrinsically before the learner measures its copies.
Standardizing that frame sends the reference-dual directions to a fixed
full-rank query and yields the symmetric-matrix model.  The authentication
reduction is already in this standardized coordinate form: its query is fixed
by the target message through $B_2$, rather than supplied by a separate
reference-frame announcement.

\begin{proposition}[Finite operational crosswalk]
\label{prop:companion-crosswalk}
Under the substitution
\begin{equation}
 (n,m,k)_{\rm learn}=(r,d,N)_{\rm auth},
 \label{eq:companion-parameter-crosswalk}
\end{equation}
and the relabeling
\begin{align}
 A_{\rm learn}&=S_{\rm auth}, & b_{\rm learn}&=e_{\rm auth},
 &J_{\rm learn}&=J_{\rm auth},\notag\\
 C_{\rm learn}&=C_{\rm auth}, &
 \gamma_{\rm learn}&=\gamma_{\rm auth},
 \label{eq:companion-label-crosswalk}
\end{align}
the residual game of
Subsection~\ref{subsec:partial-prediction-game} is exactly the odd-prime symmetric-matrix graded prediction-and-verification game of the companion paper.  In particular,
\begin{equation}
 P_{\mathrm{part}}^{(N),*}(r,d;p)
 =S_{\mathrm{ver}}^{(N),*}(r,d;p)
 =S_{\mathrm{stab,frame,ver}}^{(N),*}(r,d;p),
 \label{eq:finite-game-value-crosswalk}
\end{equation}
where the final equality is the companion paper's finite coordinate reduction from the reference-frame task to the symmetric-matrix model.
\end{proposition}

\begin{proof}
Theorem~\ref{thm:conditional-residual-state} gives the same uniform
symmetric-matrix stabilizer prior and the same tensor-power input as the learning game after
$A=S$ and $b=e$.  The true queried label is $(SJ,-J^Te)$, and admissibility is
$J^TC=C^TJ$.  If $\Delta=C-SJ$ and $t=\rank\Delta$, the common compatibility
condition is
\begin{equation}
 \gamma+J^Te\in\im((C-SJ)^T).
 \label{eq:companion-compatibility-crosswalk}
\end{equation}
Proposition~\ref{prop:partial-rank-payoff} shows that both games assign payoff
$p^{-t}$ when Eq.~\eqref{eq:companion-compatibility-crosswalk} holds and zero
otherwise.  Both allow arbitrary collective POVMs on the $N$ learner copies,
including arbitrary classical postprocessing, and both evaluate acceptance on
one additional independent copy.  Thus the first equality in
Eq.~\eqref{eq:finite-game-value-crosswalk} is a literal equality of finite
optimizations.  The second follows from the companion paper's finite coordinate reduction from the reference-frame task to the standardized symmetric-matrix model
~\cite{HayashiQuadraticStabilizerPrediction}.
\end{proof}

Combining the crosswalk with
Theorem~\ref{thm:forgery-partial-reduction} gives the finite interface used in
the converse:
\begin{equation}
 P_{\mathrm{forge}}^{(N),*}(Y\to Y')
 \le S_{\mathrm{ver}}^{(N),*}(r,d;p).
 \label{eq:forgery-to-companion-graded}
\end{equation}
The finite equivalence between the companion reference-frame and symmetric-matrix tasks and the one-directional authentication bound are logically distinct.  The former is an exact relabeling; the latter
is the protocol-specific effect-compression and randomized-completion result
needed to cover arbitrary well-formed forged outputs.

\subsection{Finite-size rank-truncated forgery bound}
\label{subsec:finite-rank-truncated-forgery}

Let $P_{\mathrm{id}}^{(N),*}(r,d;p)$ denote the optimum probability of exactly
identifying the true partial label $(SJ_d,-J_d^Te)$ from the same $N$ residual
copies.  The independent learning paper proves that the mixed-state pretty-good measurement (PGM) is optimal
for this finite exact-identification problem and gives its value by an exact
annihilator-coset formula.  For the authentication bound, only the optimum and
the following compatible-list comparison are required
~\cite{HayashiQuadraticStabilizerPrediction}.

For $0\le D\le d$, define
\begin{align}
 N_{\mathcal C}(r,d,t;p)
 &:=\left|\{B\in\mathcal C_{r,d}:\rank B=t\}\right|,
 \label{eq:admissible-rank-count-auth}\\
 L_D(r,d;p)
 &:=\sum_{t=0}^{D}N_{\mathcal C}(r,d,t;p)p^t,
 \label{eq:compatible-list-size-auth}\\
 \kappa_p^{\rm rank}
 &:=\prod_{j=1}^{\infty}(1-p^{-j}).
 \label{eq:rank-list-kappa-auth}
\end{align}
The factor $p^t$ counts the character differences compatible with a fixed
rank-$t$ chart difference.  The dimension-uniform estimate
\begin{equation}
 L_D(r,d;p)
 \le (\kappa_p^{\rm rank})^{-1}(D+1)
 p^{D(r+d+1)}
 \label{eq:compatible-list-coarse-auth}
\end{equation}
follows by dropping the linear admissibility constraint and using the standard
rectangular-rank count. Indeed, the number of $r\times d$ matrices of rank $t$ is at most $(\kappa_p^{\rm rank})^{-1}p^{t(r+d-t)}$.  After multiplication by the $p^t$ compatible character differences, the rank-$t$ contribution is at most $(\kappa_p^{\rm rank})^{-1}p^{t(r+d-t+1)}$, which is bounded by $(\kappa_p^{\rm rank})^{-1}p^{D(r+d+1)}$ for $0\le t\le D$.  Summing over $t=0,\ldots,D$ gives Eq.~\eqref{eq:compatible-list-coarse-auth}.

\begin{corollary}[General finite-copy forgery bound for a specified target]
  \label{cor:general-finite-copy-forgery-bound}
Let $Y\ne Y'$, $r=n-\ell$, and $d=\rank(Y-Y')$.  For every integer
$0\le D\le d$,
\begin{equation}
 P_{\mathrm{forge}}^{(N),*}(Y\to Y')
 \le
 L_D(r,d;p)P_{\mathrm{id}}^{(N),*}(r,d;p)
 +p^{-(D+1)}.
 \label{eq:general-finite-copy-forgery-bound}
\end{equation}
\end{corollary}

\begin{proof}
Equation~\eqref{eq:forgery-to-companion-graded} upper-bounds the forgery
probability by the graded optimum of the independent learning paper.  The companion paper's theorem ``Rank-truncated comparison of graded score with exact identification'' states, under $(n,m,k)=(r,d,N)$, that
\begin{equation}
 S_{\mathrm{ver}}^{(N),*}(r,d;p)
 \le L_D(r,d;p)P_{\mathrm{id}}^{(N),*}(r,d;p)+p^{-(D+1)}.
 \label{eq:companion-rank-truncated-import}
\end{equation}
Its proof converts the probability of the compatible rank-$D$ neighborhood into an exact-label decoder with list size $L_D$, while every outcome outside that neighborhood has score at most $p^{-(D+1)}$.  Combining this imported inequality with Eq.~\eqref{eq:forgery-to-companion-graded} gives Eq.~\eqref{eq:general-finite-copy-forgery-bound}~\cite{HayashiQuadraticStabilizerPrediction}.
\end{proof}

\begin{remark}[Division of proof responsibilities]
  \label{rem:partial-count-implicit}
The present manuscript proves the conditional residual ensemble, target-rank
identity, and authentication-specific comparison
$P_{\mathrm{forge}}\le S_{\mathrm{ver}}$.  The independent learning paper supplies
the finite exact-identification solution, compatible-list count, and
rank-truncated conversion from graded verification to exact identification.
The present proof does not use a dimension-shell relaxation or a
restricted-rank linear exponent.
\end{remark}

\subsection{Subcritical security for a specified target}
\label{subsec:partial-targets}

The exact-identification converse of the independent learning paper uses the positive function
\begin{equation}
 J(\alpha,\beta):=
 \begin{cases}
 \beta(1-\alpha-\beta/2),&0\le\alpha\le1-\beta,\\[1mm]
 (1-\alpha)^2/2,&1-\beta\le\alpha<1.
 \end{cases}
 \label{eq:authentication-J-function}
\end{equation}

\begin{table*}[t]
\caption{Finite-copy comparison and copy-rate quantities.  The authentication
parameters $(r,d,N)$ correspond to the learning ambient dimension, query
dimension, and learner-copy count.  The standard truncation
$D_r=\lfloor\sqrt r\rfloor$ is used only in the subcritical finite-size
estimate.}
\label{tab:finite-copy-comparison-quantities}
\centering
\renewcommand{\arraystretch}{1.15}
\begin{tabular}{@{}p{0.17\textwidth}p{0.25\textwidth}p{0.51\textwidth}@{}}
\hline\hline
Symbol & Definition or range & Role in the finite-copy analysis \\
\hline
$r$ & $n-\ell$ & Residual dimension after conditioning on one valid signature. \\
$d$ & $\rank(Y-Y')$, $1\le d\le r$ & Number of new residual target directions; in the linear regime, $d/r\to\beta\in(0,1]$. \\
$N$ & nonnegative integer & Number of residual public-key copies available to the adversary; the verifier uses one additional independent copy. \\
$J_d$ & $(I_d;0)\in\F_p^{r\times d}$ & Canonical full-column-rank query, obtained from any fixed query by a known invertible residual coordinate change. \\
$\mathcal C_{r,d}$ & $\{C:J_d^TC=C^TJ_d\}$ & Admissible restricted charts, equivalently the announcements defining commuting queried Weyl families. \\
$S_{\mathrm{ver}}^{(N),*}(r,d;p)$ & optimized graded score & Optimal graded score in the companion symmetric-matrix model; by the companion paper's finite coordinate reduction, it also equals the reference-relative graded score and $P_{\mathrm{part}}^{(N),*}(r,d;p)$. \\
$P_{\mathrm{id}}^{(N),*}(r,d;p)$ & optimized exact success probability & Probability of exactly identifying $(SJ_d,-J_d^Te)$ from the $N$ residual copies. \\
$D$ & integer, $0\le D\le d$ & Rank-truncation parameter separating low-rank compatible announcements from the payoff tail. \\
$N_{\mathcal C}(r,d,t;p)$ & $|\{B\in\mathcal C_{r,d}:\rank B=t\}|$ & Number of admissible restricted-chart differences in rank shell $t$. \\
$L_D(r,d;p)$ & $\sum_{t=0}^{D}N_{\mathcal C}(r,d,t;p)p^t$ & Maximum compatible list size through rank $D$; the factor $p^t$ counts compatible character differences. \\
$J(\alpha,\beta)$ & Eq.~\eqref{eq:authentication-J-function} & Positive quadratic rate supplied by the exact-identification converse of the independent learning paper for $0\le\alpha<1$ and $0<\beta\le1$. \\
$D_r$ & $\lfloor\sqrt r\rfloor$ & Standard truncation yielding an $o(r^2)$ list-size contribution and the tail $p^{-\sqrt r+O(1)}$. \\
$R_{\mathrm{auth,forge}}(p,\beta,\epsilon)$ & Eq.~\eqref{eq:authentication-forgery-rate-definition} & Minimum normalized adversarial copy rate for average forgery success for the specified target bounded asymptotically below by $\epsilon$. \\
$\alpha=1$ & copy-rate boundary & First-order threshold; finite-offset and other higher-order behavior are not specified. \\
\hline\hline
\end{tabular}
\end{table*}

\begin{corollary}[Vanishing forgery for the specified target below copy rate one]
  \label{cor:partial-target-asymptotic}
Fix an odd prime $p$.  Consider any sequence of protocol instances with
$r=n-\ell\to\infty$ such that
\begin{equation}
 d=\rank(Y-Y')=\beta r+o(r),
 \qquad
 N=\alpha r+o(r),
 \qquad
 0<\beta\le1,\quad0\le\alpha<1.
 \label{eq:authentication-partial-scaling}
\end{equation}
Then
\begin{equation}
 P_{\mathrm{forge}}^{(N),*}(Y\to Y')\longrightarrow0.
 \label{eq:authentication-subcritical-forgery}
\end{equation}
More precisely, with $D_r=\lfloor\sqrt r\rfloor$,
\begin{equation}
 P_{\mathrm{forge}}^{(N),*}(Y\to Y')
 \le
 p^{-J(\alpha,\beta)r^2+o(r^2)}
 +p^{-\sqrt r+O(1)}.
 \label{eq:authentication-subcritical-bound}
\end{equation}
\end{corollary}

\begin{proof}
Because $d/r\to\beta>0$, one has $D_r\le d$ for all sufficiently large $r$.
The exact-identification converse of the independent learning paper gives
\[
 P_{\mathrm{id}}^{(N),*}(r,d;p)
 \le p^{-J(\alpha,\beta)r^2+o(r^2)}.
\]
Moreover, Eq.~\eqref{eq:compatible-list-coarse-auth} gives
$\log_pL_{D_r}(r,d;p)=O(r^{3/2})=o(r^2)$.  Thus the list-size factor is absorbed into a new $o(r^2)$ term in the exponent.  In addition, $J(\alpha,\beta)>0$ for every $0\le\alpha<1$ and $0<\beta\le1$ by either branch of Eq.~\eqref{eq:authentication-J-function}.  Substitution into
Corollary~\ref{cor:general-finite-copy-forgery-bound} yields
Eq.~\eqref{eq:authentication-subcritical-bound}, and both terms tend to zero
~\cite{HayashiQuadraticStabilizerPrediction}.
\end{proof}

The first term in Eq.~\eqref{eq:authentication-subcritical-bound} comes from
the quadratic exact-identification converse.  The second is the tail of the
graded payoff.  Thus the displayed estimate does not claim a matching
quadratic exponent for the graded score or for authentication forgery; its
role is to prove the zero limit throughout the full region $\alpha<1$.

\subsection{Complete residual identification and direct forgery}
\label{subsec:complete-residual-direct-forgery}

The converse uses the one-directional comparison
$P_{\mathrm{forge}}\le S_{\mathrm{ver}}$.  Achievability is established by a
different, explicit protocol-level construction.  The residual ensemble is
the uniform symmetric-matrix stabilizer ensemble, so the applicable direct input is the
complete symmetric-matrix stabilizer identification theorem of the independent learning paper
~\cite{HayashiQuadraticStabilizerPrediction}.

\begin{proposition}[Complete residual identification yields an honest target signature]
\label{prop:complete-residual-yields-forgery}
Suppose that, after one exposed valid signature on $Y$, the adversary correctly
identifies the complete residual label $(S,e)$.  Then, using only the exposed
classical data and the specified target $Y'$, the adversary can compute the honest
signature $(Z_{Y'},h_{Y'})$.  Consequently,
\begin{equation}
 P_{\mathrm{forge}}^{(N),*}(Y\to Y')
 \ge P_{\mathrm{id}}^{(N),*}(r,r;p).
 \label{eq:complete-id-lower-bounds-forgery}
\end{equation}
\end{proposition}

\begin{proof}
The exposed signature determines $R_Y$ and the transformed blocks $P,Q,c$.
Together with the recovered residual pair $(S,e)$, the adversary forms
\begin{equation}
 \widetilde X=
 \begin{pmatrix}P&Q\\Q^T&S\end{pmatrix},
 \qquad
 \widetilde b=\binom{c}{e}.
 \label{eq:reconstructed-transformed-key}
\end{equation}
By the definitions of the signed-message coordinates,
\begin{equation}
 X=R_Y^{-T}\widetilde X R_Y^{-1},
 \qquad
 b=R_Y^{-T}\widetilde b.
 \label{eq:reconstructed-original-key}
\end{equation}
The adversary can therefore compute
\begin{equation}
 Z_{Y'}=XA_{Y'},
 \qquad
 h_{Y'}=-A_{Y'}^Tb,
 \label{eq:reconstructed-honest-target-signature}
\end{equation}
which is exactly the honest target signature.  Conditional on correct
residual identification, ideal correctness gives verifier acceptance with
probability one.  Running an optimal complete symmetric-matrix stabilizer identifier on the
$N$ residual copies proves Eq.~\eqref{eq:complete-id-lower-bounds-forgery}.
\end{proof}

\begin{corollary}[Supercritical direct forgery]
\label{cor:supercritical-direct-forgery}
For every sequence of specified targets, if
\begin{equation}
 N-r\longrightarrow+\infty,
 \label{eq:authentication-direct-strong-condition}
\end{equation}
then
\begin{equation}
 P_{\mathrm{forge}}^{(N),*}(Y\to Y')\longrightarrow1.
 \label{eq:authentication-supercritical-forgery}
\end{equation}
In particular, Eq.~\eqref{eq:authentication-supercritical-forgery} holds when
$N=\alpha r+o(r)$ with any fixed $\alpha>1$.
\end{corollary}

\begin{proof}
The complete symmetric-matrix stabilizer identification theorem of the independent learning paper gives
$P_{\mathrm{id}}^{(N),*}(r,r;p)\to1$ whenever $N-r\to+\infty$.
Apply Proposition~\ref{prop:complete-residual-yields-forgery} and use the
trivial upper bound one~\cite{HayashiQuadraticStabilizerPrediction}.
\end{proof}

Unlike the subcritical converse, this direct statement does not require a
positive limiting value of $d/r$.  Once the complete residual key is learned,
the honest signature can be synthesized for every specified target.

\subsection{Authentication copy rate for a specified target}
\label{subsec:authentication-copy-rate}

For $0<\beta\le1$ and $0<\epsilon<1$, define
\begin{equation}
\begin{aligned}
 &R_{\mathrm{auth,forge}}(p,\beta,\epsilon)\\
 &\quad:=\inf\Bigl\{\alpha\ge0:\
   \text{there exists a sequence of protocol instances}\\
 &\qquad \text{protocol instances such that}\\
 &\qquad r_j\to\infty,\qquad
   \frac{d_j}{r_j}\to\beta,\qquad
   \frac{N_j}{r_j}\to\alpha,\\
 &\qquad \liminf_{j\to\infty}
   P_{\mathrm{forge}}^{(N_j),*}(Y_j\to Y_j')
   \ge\epsilon\Bigr\}.
\end{aligned}
\label{eq:authentication-forgery-rate-definition}
\end{equation}
Here the $j$th instance has parameters $(n_j,\ell_j,N_j)$, with
$r_j=n_j-\ell_j$, messages
$Y_j,Y_j'\in\mathcal M_{n_j,\ell_j}$ satisfying $Y_j'\ne Y_j$, and
$d_j=\rank(Y_j-Y_j')$.  Each target is fixed independently of the hidden key
and exposed signature, and every instance is evaluated in the
one-exposed-signature, average-key experiment of
Section~\ref{sec:security-model}.
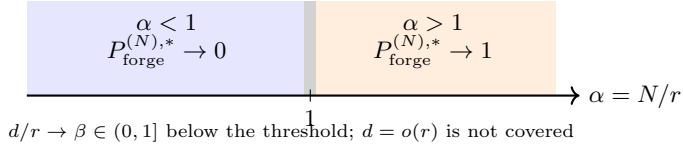
\begin{figure}[t]
\centering
\begin{tikzpicture}[x=1.25cm,y=1cm,font=\small]
\fill[blue!10] (0,0) rectangle (3.0,1.25);
\fill[orange!13] (3.0,0) rectangle (5.6,1.25);
\fill[gray!35] (2.94,0) rectangle (3.06,1.25);
\draw[->,thick] (0,0)--(5.85,0) node[right] {$\alpha=N/r$};
\draw (3,0.08)--(3,-0.08) node[below] {$1$};
\node[align=center] at (1.48,.72) {$\alpha<1$\\$P_{\mathrm{forge}}^{(N),*}\to0$};
\node[align=center] at (4.3,.72) {$\alpha>1$\\$P_{\mathrm{forge}}^{(N),*}\to1$};
\node[align=center,font=\scriptsize] at (2.8,-.48) {$d/r\to\beta\in(0,1]$ below the threshold; $d=o(r)$ is not covered};
\end{tikzpicture}
\caption{First-order copy-rate transition for forgery of a specified target. For every fixed positive target fraction $d/r\to\beta\in(0,1]$, the optimized average forgery probability vanishes when $N/r\to\alpha<1$ and tends to one when $\alpha>1$. The diagram records the copy-rate threshold; higher-order behavior around $N=r$ is not represented. The sublinear target-rank regime $d=o(r)$ requires a separate scaling analysis.}
\label{fig:copy-rate-transition}
\end{figure}
\begin{theorem}[Authentication copy rate for a specified target]
\label{thm:authentication-copy-rate-one}
For every odd prime $p$, every $0<\beta\le1$, and every
$0<\epsilon<1$,
\begin{equation}
 R_{\mathrm{auth,forge}}(p,\beta,\epsilon)=1.
 \label{eq:authentication-copy-rate-one}
\end{equation}
More strongly, for every parameter sequence satisfying
$d/r\to\beta\in(0,1]$ and $N/r\to\alpha$,
\begin{equation}
 \begin{cases}
 P_{\mathrm{forge}}^{(N),*}(Y\to Y')\to0,&\alpha<1,\\
 P_{\mathrm{forge}}^{(N),*}(Y\to Y')\to1,&\alpha>1.
 \end{cases}
 \label{eq:authentication-partial-asymptotic}
\end{equation}
The statement determines the first-order copy-rate threshold; it does not specify higher-order behavior within the window $N=r+o(r)$.
\end{theorem}

\begin{proof}
Corollary~\ref{cor:partial-target-asymptotic} proves that every sequence with
limiting copy rate below one has forgery probability tending to zero.  Hence
$R_{\mathrm{auth,forge}}\ge1$ for every $\epsilon>0$.
For the reverse inequality, choose any admissible sequence with specified targets with
$d/r\to\beta$ and take, for example,
$N=r+\lceil\sqrt r\rceil$.  Then $N/r\to1$ and $N-r\to+\infty$, so
Corollary~\ref{cor:supercritical-direct-forgery} gives forgery probability
tending to one.  Thus $R_{\mathrm{auth,forge}}\le1$.  The stronger
sequence-wise statements are exactly
Corollaries~\ref{cor:partial-target-asymptotic} and
\ref{cor:supercritical-direct-forgery}.
\end{proof}

\begin{remark}[Scope of the threshold]
  \label{rem:partial-asymptotic-interpretation}
The theorem concerns a uniform secret-key prior, one exposed valid signature,
a separately specified target, arbitrary collective processing of the finite public-key
supply, and ideal verification on one independent copy.  It does not establish
adaptive-target, existential-forgery, worst-case-key, noisy, or composable
security.  Within the first-order critical class $N/r\to1$, the regime $N-r\to+\infty$ is already covered by the direct theorem.  The finer bounded-offset regime $N=r+O(1)$ and the remaining negative or intermediate offset regimes are undetermined, as is the
sublinear target-rank regime $d=o(r)$.
\end{remark}


\section{Multiple signatures and limits of key reuse}
\label{sec:key-reuse}

The finite-copy bounds above address quantum information extracted from a
bounded supply of public-key states after one valid signature has been
exposed.  Reusing the same secret key for several messages creates a distinct
classical leakage mechanism.  Each signature reveals the action of the secret
symmetric matrix and the secret character on an additional message subspace.
This section determines exactly what can be synthesized from several valid
signatures and when the full secret key is recovered.

\subsection{The observed input subspace}
\label{subsec:observed-input-subspace}

Suppose that the same secret key $(X,b)$ is used to sign messages
\begin{equation}
  Y_1,\ldots,Y_s\in\mathcal M_{n,\ell}.
  \label{eq:multiple-signed-messages}
\end{equation}
For each message, define
\begin{equation}
  A_i:=A_{Y_i}=\begin{pmatrix}I_\ell\\Y_i\end{pmatrix},
  \qquad
  Z_i:=XA_i,
  \qquad
  h_i:=-A_i^Tb.
  \label{eq:multiple-signature-blocks}
\end{equation}
Collect the input, chart-image, and character data as
\begin{align}
  A_{\mathrm{obs}}&:=[A_1\ \cdots\ A_s],
  \qquad
  Z_{\mathrm{obs}}:=[Z_1\ \cdots\ Z_s],
  \notag\\[-2pt]
  h_{\mathrm{obs}}&:=
  \begin{pmatrix}h_1&\cdots&h_s\end{pmatrix}^{T}.
  \label{eq:observed-matrices}
\end{align}
Then
\begin{equation}
  Z_{\mathrm{obs}}=XA_{\mathrm{obs}},
  \qquad
  h_{\mathrm{obs}}=-A_{\mathrm{obs}}^Tb.
  \label{eq:observed-linear-relations}
\end{equation}
Let
\begin{align}
  W_{\mathrm{obs}}&:=\operatorname{col}(A_{\mathrm{obs}})
  \subseteq\F_p^n,
  \notag\\[-2pt]
  \rho_s&:=\dim W_{\mathrm{obs}}
  =\rank A_{\mathrm{obs}}.
  \label{eq:observed-subspace-rank}
\end{align}
The signatures determine the restriction of the linear map $X$ to
$W_{\mathrm{obs}}$ and the restriction of the linear functional
$u\mapsto-b^Tu$ to the same subspace.  Table~\ref{tab:key-reuse-notation}
collects the notation used for accumulated classical leakage and recovery.

\begin{table*}[t]
\caption{Multiple-signature and key-reuse notation.}
\label{tab:key-reuse-notation}
\centering
\renewcommand{\arraystretch}{1.15}
\begin{tabular}{@{}p{0.16\textwidth}p{0.29\textwidth}p{0.48\textwidth}@{}}
\hline\hline
Symbol & Definition & Meaning \\
\hline
$s$ & positive integer & Number of exposed valid signatures under one secret key. \\
$Y_1,\ldots,Y_s$ & signed messages & Messages authenticated using the same secret key. \\
$A_i$ & $A_{Y_i}=(I_\ell;Y_i)$ & Embedded input subspace for message $Y_i$. \\
$A_{\mathrm{obs}}$ & $[A_1\ \cdots\ A_s]$ & Cumulative observed input matrix. \\
$Z_{\mathrm{obs}}$ & $[Z_1\ \cdots\ Z_s]$ & Cumulative disclosed chart images. \\
$h_{\mathrm{obs}}$ & $(h_1^T,\ldots,h_s^T)^T$ & Cumulative disclosed character values. \\
$W_{\mathrm{obs}}$ & $\operatorname{col}(A_{\mathrm{obs}})$ & Input subspace on which the chart and character are exposed. \\
$\rho_s$ & $\rank A_{\mathrm{obs}}$ & Dimension of the observed input span. \\
$\Delta_s$ & $[Y_2-Y_1\ \cdots\ Y_s-Y_1]$ & Concatenated message differences; $\rho_s=\ell+\rank\Delta_s$. \\
$A_\star$ & target message embedding & Classically forgeable when $\operatorname{col}(A_\star)\subseteq W_{\mathrm{obs}}$. \\
$T$ & $A_\star=A_{\mathrm{obs}}T$ & Coefficient matrix used to synthesize a target signature. \\
$R$ & $A_{\mathrm{obs}}R=I_n$ & Right inverse used when the observed span has full rank. \\
$X$ & $Z_{\mathrm{obs}}R$ & Recovered secret chart in the full-rank case. \\
$b$ & $-R^Th_{\mathrm{obs}}$ & Recovered secret character parameter in the full-rank case. \\
$r_s$ & $n-\rho_s=n-\ell-\rank\Delta_s$ & Residual hidden dimension after $s$ signatures. \\
\hline\hline
\end{tabular}
\end{table*}

\subsection{Exact synthesis on the observed span}
\label{subsec:span-synthesis}

\begin{proposition}[Classical synthesis criterion]
  \label{prop:classical-synthesis-criterion}
  Let $Y_\star\in\mathcal M_{n,\ell}$ be a target message and set
  $A_\star:=A_{Y_\star}$.  If
  \begin{equation}
    \operatorname{col}(A_\star)\subseteq W_{\mathrm{obs}},
    \label{eq:target-contained-observed-span}
  \end{equation}
  then an observer of the $s$ valid signatures can compute a valid signature
  for $Y_\star$ using only classical linear algebra, with no public-key copy.
\end{proposition}

\begin{proof}
  Condition~\eqref{eq:target-contained-observed-span} is equivalent to the
  existence of a matrix
  \begin{equation}
    T\in\F_p^{s\ell\times\ell}
    \quad\text{such that}\quad
    A_\star=A_{\mathrm{obs}}T.
    \label{eq:target-coefficient-matrix}
  \end{equation}
  Define
  \begin{equation}
    Z_\star:=Z_{\mathrm{obs}}T,
    \qquad
    h_\star:=T^Th_{\mathrm{obs}}.
    \label{eq:synthesized-signature}
  \end{equation}
  Using Eq.~\eqref{eq:observed-linear-relations},
  \begin{equation}
    Z_\star=XA_{\mathrm{obs}}T=XA_\star,
    \qquad
    h_\star=-T^TA_{\mathrm{obs}}^Tb=-A_\star^Tb.
    \label{eq:synthesized-signature-correctness}
  \end{equation}
  Thus $(Z_\star,h_\star)$ is exactly the honest signature generated with the
  secret key, and perfect correctness implies acceptance with probability one.
\end{proof}

\begin{remark}[The leakage criterion is a span condition]
  \label{rem:span-condition-not-message-count}
  The decisive quantity is $\rho_s=\rank A_{\mathrm{obs}}$, not the number
  $s$ of signatures by itself.  Repeated or linearly redundant messages need
  not enlarge $W_{\mathrm{obs}}$, whereas carefully selected messages can add
  several new directions.
\end{remark}

\subsection{Rank growth for graph messages}
\label{subsec:rank-growth-graph-messages}

Fix the first signed message and define
\begin{equation}
  \Delta_s
  :=[Y_2-Y_1\ \cdots\ Y_s-Y_1]
  \in\F_p^{(n-\ell)\times(s-1)\ell}.
  \label{eq:message-difference-concatenation}
\end{equation}

\begin{lemma}[Observed-rank identity]
  \label{lem:observed-rank-identity}
  The accumulated input matrix satisfies
  \begin{equation}
    \rho_s
    =\rank A_{\mathrm{obs}}
    =\ell+\rank\Delta_s.
    \label{eq:observed-rank-formula}
  \end{equation}
  Consequently,
  \begin{equation}
    \ell\le\rho_s\le
    \min\{n,s\ell\}.
    \label{eq:observed-rank-bounds}
  \end{equation}
\end{lemma}

\begin{proof}
  For each $i\ge2$, subtract the first block column $A_1$ from $A_i$.
  This invertible block-column operation transforms $A_{\mathrm{obs}}$ into
  \begin{equation}
    \begin{bmatrix}
      I_\ell&0&\cdots&0\\
      Y_1&Y_2-Y_1&\cdots&Y_s-Y_1
    \end{bmatrix}.
    \label{eq:observed-rank-reduced-matrix}
  \end{equation}
  The first $\ell$ columns have independent top coordinates, while every
  remaining column has zero top block.  Their additional rank is exactly
  $\rank\Delta_s$, proving Eq.~\eqref{eq:observed-rank-formula}.  The bounds
  follow from matrix dimensions.
\end{proof}

For a target $Y_\star$, column reduction relative to $Y_1$ also gives
\begin{equation}
  \operatorname{col}(A_{Y_\star})\subseteq W_{\mathrm{obs}}
  \quad\Longleftrightarrow\quad
  \operatorname{col}(Y_\star-Y_1)
  \subseteq\operatorname{col}(\Delta_s).
  \label{eq:target-message-difference-criterion}
\end{equation}
Thus the already forgeable message family can be characterized entirely in
terms of the span of the observed message differences.

\subsection{Full secret-key recovery}
\label{subsec:full-key-recovery}

\begin{theorem}[Full key recovery from spanning signatures]
  \label{thm:full-key-recovery}
  If
  \begin{equation}
    \rank A_{\mathrm{obs}}=n,
    \label{eq:full-rank-observed-condition}
  \end{equation}
  then the valid classical signatures determine both $X$ and $b$ uniquely.
  In particular, any message can subsequently be signed with certainty without
  using a quantum public-key copy.
\end{theorem}

\begin{proof}
  Full row rank of $A_{\mathrm{obs}}\in\F_p^{n\times s\ell}$ implies the
  existence of a right inverse
  \begin{equation}
    R\in\F_p^{s\ell\times n}
    \quad\text{such that}\quad
    A_{\mathrm{obs}}R=I_n.
    \label{eq:Aobs-right-inverse}
  \end{equation}
  Therefore
  \begin{equation}
    X=Z_{\mathrm{obs}}R.
    \label{eq:X-full-recovery}
  \end{equation}
  Moreover, transposing Eq.~\eqref{eq:Aobs-right-inverse} gives
  $R^TA_{\mathrm{obs}}^T=I_n$, and hence
  \begin{equation}
    b=-R^Th_{\mathrm{obs}}.
    \label{eq:b-full-recovery}
  \end{equation}
  Both secret parameters are thus recovered.  The honest signing equations can
  then be evaluated for every target message.
\end{proof}

Combining Theorem~\ref{thm:full-key-recovery} with
Lemma~\ref{lem:observed-rank-identity}, full recovery occurs exactly when
\begin{equation}
  \rank\Delta_s=n-\ell=r.
  \label{eq:full-key-recovery-difference-condition}
\end{equation}
A necessary signature-count condition is
\begin{equation}
  s\ell\ge n,
  \qquad\text{equivalently}\qquad
  s\ge\left\lceil\frac n\ell\right\rceil.
  \label{eq:necessary-signature-count}
\end{equation}
This condition is also achievable under chosen-message access: the signer can
be queried on messages whose difference blocks jointly span
$\F_p^{n-\ell}$.  Hence
$\lceil n/\ell\rceil$ suitably selected signatures suffice for full key
recovery.  The number may be smaller than a naive parameter count would
suggest because each signature reveals the chart action and character on an
entire $\ell$-dimensional subspace.

\subsection{Residual uncertainty after several signatures}
\label{subsec:multi-signature-residual-uncertainty}

The same canonical-coordinate argument used for one signature extends to the
observed span $W_{\mathrm{obs}}$.  Choose a basis whose first $\rho_s$
coordinates span $W_{\mathrm{obs}}$.  The signatures reveal the corresponding
columns of the transformed symmetric chart and the character on those
coordinates.  By symmetry, the matching rows are also known.  The remaining
unknown secret is a uniformly distributed quadratic stabilizer parameter on
\begin{equation}
  r_s:=n-\rho_s
  =n-\ell-\rank\Delta_s
  \label{eq:multi-signature-residual-dimension}
\end{equation}
coordinates.

\begin{proposition}[Residual ensemble after multiple exposures]
  \label{prop:multi-signature-residual-ensemble}
  Under the uniform key prior, conditioned on $s$ valid signatures, the
  unrevealed chart block and character vector remain independent and uniform
  on
  \begin{equation}
    \Sym_{r_s}(\F_p)\times\F_p^{r_s}.
    \label{eq:multi-signature-residual-parameter-space}
  \end{equation}
  A known unitary maps every public-key copy to a known factor tensored with a
  standard $r_s$-qudit quadratic stabilizer state.
\end{proposition}

\begin{proof}
  Complete a basis of $W_{\mathrm{obs}}$ to a basis of $\F_p^n$ and apply the
  corresponding symplectic coordinate change.  In block form, the observed
  signatures fix the first $\rho_s$ chart columns and the first $\rho_s$
  character coordinates.  Symmetry fixes the corresponding chart rows.  The
  remaining lower-right symmetric block and residual character are untouched.
  Uniformity and independence follow from the same bijective block-coordinate
  argument as in Lemma~\ref{lem:signature-exposes-blocks}.  Removing the known
  phase terms gives the stated tensor factorization, exactly as in
  Theorem~\ref{thm:conditional-residual-state}.
\end{proof}

This proposition clarifies the separation between two effects.  Classical
signature exposure reduces the effective dimension from $n$ to $r_s$, while
quantum-public-key copies may help predict the remaining residual parameter.
A bounded-use extension of the finite-copy theorem would therefore replace the
single-signature residual dimension $r=n-\ell$ by $r_s$ and analyze only target
directions outside $W_{\mathrm{obs}}$.  Such an extension requires an explicit
multi-signature security game and is not claimed by the present main theorem.

\subsection{Security scope and deployment implication}
\label{subsec:key-reuse-security-scope}

Theorem~\ref{thm:full-key-recovery} rules out interpreting the present scheme
as an unrestricted reusable signature scheme.  In particular, the finite-copy
quantum security upper bound does not prevent classical key recovery after sufficiently
many linearly independent signatures.  The security results of
Secs.~\ref{sec:security-model}--\ref{sec:finite-copy-security} are therefore
stated for one exposed valid signature.

A deployment based on the present construction must use one of the following
restrictions:
\begin{enumerate}[label=(\roman*),leftmargin=*]
  \item use each secret key for only one authenticated message;
  \item impose and enforce a rank budget on the cumulative observed subspace;
  \item refresh the secret key and quantum public-key supply before the
        observed span becomes too large; or
  \item redesign the classical signature so that it does not reveal the
        unrestricted linear images $XA_Y$ and $A_Y^Tb$.
\end{enumerate}
The present paper proves the first option's specified target,
single-signature-exposure security statement and identifies the exact algebraic
obstruction to unrestricted reuse.  The adjective ``one-time'' is used only as
a deployment restriction on secret-key use; it is not a claim of existential
one-time unforgeability over adaptively selected messages.  The paper does not
claim EUF-CMA security.

\section{Discussion}
\label{sec:discussion}

The result of this paper is best understood by separating four neighboring
research directions that are often grouped under the label quantum digital
signatures: standardized post-quantum signatures, communication-oriented
multipartite QDS, computational signatures with quantum cryptographic objects,
and information-theoretic schemes based on a finite supply of quantum public
keys.  These directions differ not only in implementation but also in who
verifies, what resource is quantum, how many signing queries are permitted,
and whether the adversary is computationally bounded.  The present protocol
belongs to the last direction.  Its contribution is a sharp off-boundary
copy-rate transition for forgery for the specified target against collective attacks on a
structured stabilizer public key, together with an explicit account of the
classical information leaked by valid signatures.  It is not a replacement
for post-quantum signature standards and does not inherit the transferability
or multi-use guarantees pursued in other QDS models.

\subsection{Relation to post-quantum and computational signatures}
\label{subsec:discussion-computational-signatures}

Post-quantum signatures such as ML-DSA and SLH-DSA use freely copyable
classical public keys and base security on computational hardness
\cite{NISTFIPS204,NISTFIPS205}.  Their operational objective is the standard
one: efficient public verification under many uses and polynomial-time attack.
By contrast, a quantum public key is a consumable physical resource whose
copy count is part of the security model.  The present scheme exchanges the
standard computational assumption for a strict supply assumption and a much
narrower signing interface.  Accordingly, its information-theoretic security upper bound
should not be interpreted as a stronger version of a standard post-quantum
signature.  It answers a different question: how well can even an unbounded
adversary forge after receiving a specified number of copies and one valid
signature?

The same distinction applies to computational quantum cryptography.  One-way
state generators and related quantum assumptions support commitments,
bounded-time-secure quantum-public-key signatures, and other primitives
\cite{MorimaeYamakawa2022,MorimaeYamakawa2024,KhuranaTomer2024}.  Those works
restrict the adversary to quantum polynomial time but can formulate security
through asymptotic computational indistinguishability and can seek multi-use
notions.  Black-box separation results also show that a one-time construction
from a quantum-state assumption does not automatically yield a standard
multi-time signature~\cite{ColadangeloMutreja2024}.  This point is directly
relevant here.  Our single-signature theorem is information theoretic with
respect to quantum computation, but the key-reuse theorem prevents it from
being promoted to EUF-CMA merely by changing the adversary's running-time
bound.

A useful conceptual distinction is therefore between two kinds of one-wayness.
Computational quantum one-wayness asks whether an efficient algorithm can
invert or predict a quantum-state generator despite having polynomially many
resources.  Finite-copy statistical one-wayness asks how much any measurement,
regardless of complexity, can infer from a bounded physical sample.  The
present rank-shell analysis is of the second type.  The two notions can be
combined in future constructions, but neither subsumes the other in the
security experiments presently available.

\subsection{Relation to communication-oriented quantum digital signatures}
\label{subsec:discussion-communication-qds}

Multipartite QDS has made substantial progress toward practical quantum
networks.  Measurement-device-independent, continuous-variable, long-distance,
and integrated implementations address detector vulnerabilities, signing
rate, transmission loss, and deployment on shared optical infrastructure
\cite{Puthoor2016,Roberts2017,Richter2021,Du2025}.  Network demonstrations also
combine signatures with encryption or other communication services
\cite{Yin2023}.  These achievements concern requirements that the present
paper leaves open, including transferability between recipients,
non-repudiation, robustness under channel noise, and experimentally defined
acceptance thresholds.

The present model instead gives each verifier a physical copy of a quantum
public key and asks whether that verifier can independently test a classical
message--signature pair.  This architecture avoids assigning a participant a
role in the act of verification, but it shifts the burden to authenticated
state distribution, coherent storage, and explicit copy accounting.  Thus a
fair comparison cannot be made using signing rate or distance alone.  It must
also include the recipient structure, the function of symmetrization or
secret sharing, the number and lifetime of stored quantum states, the
available signing queries, and the security definition.

In particular, the information-theoretic terminology used in experimental QDS
and in the present finite-copy theorem refers to different complete systems.
Communication-oriented protocols typically derive their security from a
combination of quantum-state distribution, sampling, and classical
post-processing in a multipartite setting.  Our security upper bound derives from the
limited distinguishability of a structured state ensemble after classical
signature leakage has been conditioned upon.  Neither analysis can be
transferred to the other architecture without a new reduction.

\subsection{Comparison with the single-qubit finite-copy protocol}
\label{subsec:comparison-single-qubit-qds}
The single-qubit protocol of Wang and Hayashi is the closest direct comparator
because it also uses a classical secret description, quantum public-key
states, independent verification, and a computationally unbounded adversary
with a bounded number of copies~\cite{WangHayashi2026QDS}.  For each message
bit and repetition index, its public key contains a state from
$\{\ket{0},\ket{1},\ket{+},\ket{-}\}$.  The signature reveals the basis and
eigenvalue corresponding to the selected bit value, and verification uses
independent single-qubit measurements.  For message length $L$ and repetition
parameter $\lambda$, one public-key copy contains $2\lambda L$ qubits.

Our protocol instead stores one quadratic stabilizer state on $n$ qudits.  A
message $Y\in\F_p^{(n-\ell)\times\ell}$ selects a subspace of commuting Weyl
operators, and the signature gives the restriction of the graph chart and
character to that subspace.  Verification is a joint spectral test.  The
public-key state may be entangled, and preparation and verification can
require multi-qudit Clifford control.  The present construction therefore
gives up the especially simple local implementation of the single-qubit
scheme in order to access a high-dimensional algebraic prediction problem.

The two security mechanisms are structurally different.  In the single-qubit
protocol, changing a bit forces the adversary to predict $\lambda$
independently sampled elementary labels.  With no public-key copy the
corresponding factor is $2^{-\lambda}$, while with $K$ copies the exact
four-state estimation probability gives a bound of the form $F(K)^\lambda$
~\cite{WangHayashi2026QDS}.  Independent repetition is therefore the direct
amplification resource.

Here, one exposed signature leaves a residual stabilizer ensemble of dimension
\begin{equation}
  r=n-\ell,
  \label{eq:discussion-residual-dimension}
\end{equation}
and a target message introduces
\begin{equation}
  d=\rank(Y-Y')
  \label{eq:discussion-target-rank}
\end{equation}
new residual directions.  For every fixed positive limiting target-rank
fraction $d/r\to\beta\in(0,1]$, the forgery probability for the specified target tends to
zero when $N/r\to\alpha<1$ and tends to one when $N/r\to\alpha>1$.  Thus the
relevant asymptotic resource is the number of adversarial copies per residual
dimension, and the transition occurs at copy rate one away from the critical
line.  This conclusion comes from the residual graded-verification converse
below the threshold and complete residual-state identification above it, not
from an independent repetition argument.

The two results should not be converted into a uniform dominance claim.  The
single-qubit bound is an exact elementary-state estimation expression raised
to a repetition parameter, whereas the present theorem concerns a
high-dimensional residual ensemble, a specified target whose rank grows linearly
with the residual dimension, and a first-order copy-rate threshold whose finer window is not evaluated here.
Moreover, targets with $d=o(r)$, including rank-one changes, remain outside the
subcritical copy-rate theorem.  These differences in target structure and
asymptotic scaling prevent a direct comparison of the numerical decay laws.

The implementation tradeoff is equally important.  The single-qubit scheme
requires only product-state preparation and $X/Z$ measurements, but its quantum
key grows through explicit repetition.  Our public-key copy uses $n$ qudits,
but the classical signature contains $(n+1)\ell$ field elements before any
compression, and generic preparation or measurement can be substantially more
complex.  Consequently, the present paper does not claim a practical resource
advantage.  A meaningful comparison would need state-preparation depth,
quantum-memory time, local dimension, public-key distribution rate, classical
signature size, target structure, and permitted key reuse in addition to the
forgery bound.
\subsection{Adjacent QPUF authentication}
\label{subsec:discussion-qpuf}

Quantum physical unclonable functions provide another route to authentication
through an unknown physical transformation or measurement device
\cite{ArapinisEtAl2021,GhoshEtAl2024}.  The conceptual overlap with this paper
is the attempt to convert physical nonclonability and restricted observations
into a forgery bound.  The operational object is nevertheless different.  The
present signer produces a classical signature for a freely chosen classical
message, and verification uses a distributed state associated with a
classical secret description.  Many QPUF protocols instead implement
challenge--response identification of a physical object and do not produce a
classical message signature with the same verification interface.

This neighboring literature is useful in two ways.  First, it reinforces that
no-cloning alone is not a security proof: the learnability of the challenge
family, the measurement model, and the number of observations must all be
specified.  Second, it suggests hardware-oriented alternatives to storing a
large distributed public-key state.  Establishing a formal reduction between
the stabilizer prediction game and a realizable QPUF would, however, require a
new model and is outside the present claims.

\subsection{Security scope and key reuse}
\label{subsec:comparison-security-scope}
The main theorem concerns a distinct target specified independently of the hidden key and exposed signature, averages over the uniform
key prior, permits an arbitrary collective attack on $N$ copies, and assumes
that one valid signature has already been exposed.  It does not imply
adaptive-target, existential-forgery, worst-case-key, non-repudiation,
transferability, or composable security.  For target sequences with
$d/r\to\beta\in(0,1]$, the theorem establishes the zero--one transition away
from copy rate one.  It makes no pointwise claim when $N/r\to1$.  Targets with
$d=o(r)$, including rank-one changes, also lie outside the subcritical
converse and require a separate scaling analysis.

The classical leakage analysis further shows why bounded public-key copies and
bounded signing queries must be treated as independent resources.  Each valid
signature reveals $XA_Y$ and $A_Y^Tb$.  Several signatures therefore reveal
the secret linear map and character on the cumulative span of the message
subspaces.  A target contained in that span can be signed by classical linear
combination, and full observed rank recovers the complete secret key.  Under
chosen-message access, $\lceil n/\ell\rceil$ suitably selected signatures
suffice.  This obstruction does not depend on the adversary's ability to
measure the quantum public key and cannot be repaired by strengthening the
finite-copy prediction theorem.

The single-qubit protocol has a different, coordinatewise exposure pattern:
a signature reveals the secret branch corresponding to each signed bit, while
the branch for the opposite bit remains unrevealed.  That difference prevents
the present span-recovery theorem from being transferred directly to the
single-qubit scheme, but it also means that repeated signing requires a
separate analysis there.  More generally, a bounded-copy theorem should never
be presented as chosen-message security unless classical signature exposure
has been incorporated into the experiment.
\subsection{Physical assumptions and implementation}
\label{subsec:resource-tradeoffs}

Both the present scheme and the single-qubit finite-copy scheme assume
authenticated distribution of intact public-key states.  If an adversary can
replace the verifier's designated test state, finite-copy distinguishability
is no longer the relevant security mechanism.  Both analyses also idealize
state preparation, memory, and measurement.  The present joint stabilizer test
would need a noise-tolerant threshold, while the security proof would have to
bound an adversary who exploits the same threshold.  Such a result requires a
joint completeness--soundness analysis rather than simply adding an honest
error probability.

The public-key copy budget must also be global.  Copies assigned to adversarial
access, honest verification, calibration, failed retrieval, and repeated
verification all draw from the same physical supply.  Ideal projective
verification may preserve an honest eigenstate conditionally on acceptance,
but that observation alone does not establish composable reuse of the same
copy.  The present model therefore reserves an independent test copy for each
verification instance.

\subsection{Open mathematical problems}
\label{subsec:discussion-limitations-open-problems}
The first problem is the critical threshold window.  The present results already cover sequences with $N-r\to+\infty$, even though such sequences may satisfy $N/r\to1$.  What remains unresolved is the finer behavior when the overhead does not diverge, especially $N=r+O(1)$, as well as the corresponding negative and intermediate offset regimes.
The second problem is the sublinear-query regime.  The subcritical converse
assumes $d/r\to\beta$ with fixed $\beta>0$ and therefore does not cover fixed
$d$, general $d=o(r)$, or intermediate scalings in which $d\to\infty$ but
$d/r\to0$.  These regimes include rank-one and other low-rank target
modifications that are central to the distinction between a separately specified target and
existential security.

The third problem is the finite graded optimum.  The independent learning paper
solves the finite exact partial-label identification problem through the
mixed-state PGM and an exact formula organized by annihilator cosets, namely cosets fixed by the relevant orthogonality constraints, but this does not by
itself give a closed solution of the finite graded-verification SDP.  Solving
that SDP directly could reveal finite-size behavior hidden by the
rank-truncated comparison and could clarify higher-order threshold behavior.

The fourth problem is refined graded decay below copy rate one.  The present
proof combines the quadratic exact-identification converse with the truncation
$D_r=\lfloor\sqrt r\rfloor$ and obtains
\begin{equation}
  p^{-J(\alpha,\beta)r^2+o(r^2)}
  +p^{-\sqrt r+O(1)}.
  \label{eq:open-rank-truncated-decay}
\end{equation}
This estimate proves convergence to zero throughout $\alpha<1$, but the tail
term need not describe the true graded asymptotics.  Determining the optimal
decay scale and its dependence on $(\alpha,\beta,p)$ remains open.

The fifth problem is bounded-use security.  After several signatures the
unrevealed state is again a quadratic stabilizer ensemble, now on
$r_s=n-\rank A_{\mathrm{obs}}$ residual coordinates.  This suggests a theorem
that tracks both the cumulative classical rank budget and the quantum copy
rate.  Proving it requires an adaptive multi-signature experiment, precise
accounting of target directions outside the evolving observed span, and a
partial-prediction result uniform over that evolution.

The sixth problem is a qubit authentication protocol.  The independent learning
theorem covers all prime local dimensions, including $p=2$, but the present
authentication protocol uses an odd-prime Weyl convention and the inverse of
two in its quadratic phases.  Extending the authentication theorem therefore
requires a qubit signature convention based on an enhanced quadratic phase or
an equivalent $\mathbb Z_4$ refinement, together with qubit versions of the
conditional residual reduction, verification projector, and effect-compression
argument.  This is an authentication-side extension rather than a missing
qubit learning theorem.

A further direction is to combine information-theoretic finite-copy analysis
with computational assumptions.  Computational one-wayness could protect a
residual parameter after the information-theoretic guarantee becomes weak,
while the physical copy limit could reduce the adversary's effective input
before a computational reduction is applied.  Such a hybrid statement would
need a single security experiment that accounts simultaneously for copy
supply, running time, and signing queries.
\subsection{Conclusion}
\label{subsec:discussion-conclusion}
This paper shows that stabilizer-public-key authentication can be analyzed as
a conditional structured-prediction problem.  After one valid signature, a
known unitary separates every public-key copy into an exposed factor and an
exactly uniform quadratic stabilizer ensemble on $r=n-\ell$ residual qudits.
A specified target at rank distance $d=\rank(Y-Y')$ asks only for the residual
chart action and character on its $d$ new directions.  Even when a forged
commuting family couples the exposed and residual factors, compression against
the known factor and randomized isotropic completion reduce its retained
verification effect to the same graded-verification game for a specified query.

This reduction yields a sharp off-boundary finite-copy transition.  For every
fixed positive limiting target-rank fraction $d/r\to\beta\in(0,1]$, the optimal
average forgery probability for the specified target tends to zero whenever
$N/r\to\alpha<1$.  Conversely, complete identification of the residual
stabilizer label succeeds with probability tending to one when
$N-r\to+\infty$.  From the identified residual label and the exposed classical
blocks, the adversary reconstructs the secret chart and character and outputs
the honest target signature itself.  Thus the forgery probability tends to
one for every $\alpha>1$, and the authentication forgery copy rate for the
specified target equals one.  No pointwise conclusion is made on the critical line
$N/r\to1$, and the subcritical theorem does not cover $d=o(r)$.

The quantum copy threshold and classical key reuse are independent
constraints.  Repeated valid signatures reveal the secret chart and character
on their cumulative input span.  Every target whose message subspace lies in
that span can be signed by classical linear combination, and full observed
rank recovers the complete secret key.  The scheme must therefore be treated
as one-time or explicitly rank-budgeted authentication rather than as a
conventional reusable signature scheme.

The resulting guarantee has a precise scope.  It is information theoretic
against arbitrary collective attacks, but it concerns a separately specified target, uses an average over keys,
based on one exposed valid signature, and formulated for ideal operations with
an independently allocated verifier copy and authenticated public-key
distribution.  It does not establish adaptive-target or existential
unforgeability, worst-case-key security, EUF-CMA security, safe verification-
copy reuse, noise tolerance, or composable security.  Within this restricted
interface, the result isolates a concrete role for high-dimensional
stabilizer geometry: below one adversarial copy per residual dimension,
fixed positive-rank target prediction becomes asymptotically impossible,
whereas above that rate complete residual learning enables exact forgery.
Refining the threshold window and analyzing sublinear target-rank regimes, proving a
rank-budgeted multi-use theorem, extending the authentication construction to
qubits, and developing noise-tolerant verification are the principal next
steps.

\subsection*{Acknowledgement}
M.H. and B.Y. was supported in part by
the Guangdong Provincial Quantum Science Strategic Initiative (Grant No. GDZX2505003).
M.H. was supported in part by
the General R\&D Projects of 1+1+1 CUHK-CUHK(SZ)-GDST Joint Collaboration Fund (Grant No. GRDP2025-022).

\appendix

\section{Coordinate reduction and marginal deletion}
\label{app:canonical-reduction}

This appendix supplies details used in
Theorem~\ref{thm:forgery-partial-reduction}.

\subsection{Symplectic coordinate change}

For $R\in\operatorname{GL}(n,\F_p)$, the matrix
\begin{equation}
  \mathsf S_R:=\begin{pmatrix}R^{-1}&0\\0&R^T\end{pmatrix}
  \label{eq:appendix-SR}
\end{equation}
is symplectic because it preserves the form in
Eq.~\eqref{eq:symplectic-form}.  Moreover,
\begin{equation}
  \mathsf S_R\binom{u}{Xu}
  =\binom{R^{-1}u}{R^TXR(R^{-1}u)}.
  \label{eq:appendix-graph-transform}
\end{equation}
Thus graph charts transform by congruence, and characters transform by
$b\mapsto R^Tb$.

\subsection{Normal form of the target quotient map}

Let $B_2\in\F_p^{r\times\ell}$ have rank $d$.  There exist
$U\in\operatorname{GL}(r,\F_p)$ and
$V\in\operatorname{GL}(\ell,\F_p)$ such that
\begin{equation}
  UB_2V=
  \begin{pmatrix}
    I_d&0\\
    0&0
  \end{pmatrix}.
  \label{eq:B2-rank-normal-form}
\end{equation}
Changing the target generator basis by $V$ does not change the joint
verification projector; it merely applies the corresponding invertible linear
map to the signed outcome vector.  The residual coordinate change $U$ is
implemented by a known Clifford unitary and preserves the uniform residual
ensemble.  Consequently, the partial game may use the canonical query matrix
\begin{equation}
  J_d:=\begin{pmatrix}I_d\\0\end{pmatrix}
  \in\F_p^{r\times d}.
  \label{eq:canonical-query-Jd}
\end{equation}

\subsection{Explicit retained blocks and kernel compression}
In the signed-message coordinates, write $B=(B_1;B_2)$ and let a purported
signature have transformed bottom block
\begin{equation}
  \widehat T:=R_Y^T\widehat Z.
  \label{eq:appendix-That}
\end{equation}
Choose $V\in\operatorname{GL}(\ell,\F_p)$ so that
\begin{equation}
  B_2V=\bigl[J\ \ 0\bigr],\qquad\rank J=d.
  \label{eq:appendix-new-old-split}
\end{equation}
After the known conjugation, retain the first $d$ generators and split their
transformed labels according to the known and residual tensor factors as
\begin{equation}
  \left(
    \binom{L}{J},
    \binom{F}{C}
  \right).
  \label{eq:appendix-retained-blocks}
\end{equation}
The matrices $L,F,C$ are obtained by applying the same known linear label
transformation to the first $d$ columns of $(BV,\widehat T V)$.  With
$K_0=\left(\begin{smallmatrix}P&Q\\Q^T&0\end{smallmatrix}\right)$ and
$c_0=(c^T,0)^T$, the diagonal conjugation obeys
\begin{equation}
  D_{P,Q,c}W(u_j,z_j)D_{P,Q,c}^{\dagger}
  =\omega^{-c^TL_j}W(u_j,z_j-K_0u_j).
  \label{eq:appendix-retained-conjugation-phase}
\end{equation}
Thus the known phase vector has entries $\kappa_j=-c^TL_j$ and is included
in the transformed character as $\eta=\widehat h_d-\kappa$.  The
well-formedness condition for the retained total family is precisely
\begin{equation}
  F^TL-L^TF+C^TJ-J^TC=0.
  \label{eq:appendix-total-isotropy}
\end{equation}
This identity does not imply factorwise isotropy.

If $\Pi_1,\ldots,\Pi_\ell$ are the transformed commuting spectral projectors,
then marginal deletion gives
\begin{equation}
  \Pi_{\mathrm{full}}
  =\prod_{j=1}^{\ell}\Pi_j
  \leq\prod_{j=1}^{d}\Pi_j
  =:\Pi_{\mathrm{ret},\eta}.
  \label{eq:appendix-projector-deletion}
\end{equation}
The retained total projector has the Fourier form in
Eq.~\eqref{eq:retained-Fourier-projector}.  Compressing it against
$\ket{+}^{\otimes\ell}$ removes every term outside $\ker F$.  On the surviving
kernel, Eq.~\eqref{eq:appendix-total-isotropy} reduces to residual isotropy,
as shown in Lemma~\ref{lem:kernel-residual-isotropy}.  Thus no separate joint
spectral measurement of the full known-factor or residual-factor projections
is assumed.

\subsection{Canonical construction of the isotropic completion}
Let $T$ be a basis matrix for $\ker F$, complete it to
$V=[T\ R]\in\operatorname{GL}(d,\F_p)$, and make an invertible residual
coordinate change that sends $JV$ to $J_d=(I_d;0)$.  In these coordinates an
admissible chart is $C=(H;R_0)$ with $H=H^T$.  Write the prescribed first $k$
columns of the upper block as
\begin{equation}
  \binom{H_{11}}{H_{21}},
  \qquad H_{11}\in\F_p^{k\times k}.
  \label{eq:appendix-prescribed-upper-block}
\end{equation}
Kernel isotropy gives $H_{11}=H_{11}^T$.  The canonical choice
\begin{equation}
  H=\begin{pmatrix}H_{11}&H_{21}^T\\H_{21}&0\end{pmatrix}
  \label{eq:appendix-symmetric-completion}
\end{equation}
therefore extends the prescribed columns to a symmetric $d\times d$ block.
Choose every unspecified column of $R_0$ to be zero and transport the result
back to the original coordinates.  This proves
Lemma~\ref{lem:isotropic-query-completion} constructively.

\subsection{Fourier proof of the random-character marginal}
In the completed generator basis $V=[T\ R]$, write a Fourier coefficient as
$(y,z)\in\F_p^k\oplus\F_p^{d-k}$.  Averaging the full character projectors over
$\zeta\in\F_p^{d-k}$ introduces the factor
\begin{equation}
  \frac1{p^{d-k}}\sum_{\zeta\in\F_p^{d-k}}
  \omega^{-\zeta^Tz}=\delta_{z,0}.
  \label{eq:appendix-character-orthogonality}
\end{equation}
Only terms with $z=0$ survive, leaving the Fourier projector of the first $k$
generators together with the normalization $p^{-(d-k)}$.  This is exactly
Eq.~\eqref{eq:random-character-projector-identity}.

\subsection{Parameter count}
For $J_d=(I_d;0)$, an admissible $C$ has a symmetric upper $d\times d$ block
and an arbitrary lower $(r-d)\times d$ block.  Thus the restricted chart has
\begin{equation}
  \frac{d(d+1)}2+(r-d)d
  =rd-\frac{d(d-1)}2
  \label{eq:partial-chart-parameter-count}
\end{equation}
field parameters.  Including the $d$ character values, the partial label space
has cardinality
\begin{equation}
  p^{(r+1)d-d(d-1)/2}.
  \label{eq:partial-label-cardinality}
\end{equation}
This count agrees with the quotient-parameter exponent obtained by direct
parameter counting, but the present derivation does not rely on that separate
counting argument.

\section{Verification measurement}
\label{app:verification-measurement}

This appendix records the spectral facts used in the protocol and correctness
proof.

\subsection{Spectral projectors of a Weyl operator}

For nonzero $g\in\F_p^{2n}$, the phase-adjusted Weyl operator satisfies
$W(g)^p=I$.  Its eigenvalues are therefore among
$\{\omega^c:c\in\F_p\}$.  For $c\in\F_p$, define
\begin{equation}
  \Pi_{g,c}:=\frac1p\sum_{s\in\F_p}\omega^{-sc}W(g)^s.
  \label{eq:appendix-Weyl-projector}
\end{equation}
Then
\begin{align}
  W(g)\Pi_{g,c}
  &=\frac1p\sum_{s\in\F_p}\omega^{-sc}W(g)^{s+1}
    \notag\\
  &=\omega^c\Pi_{g,c}.
  \label{eq:appendix-projector-eigenvalue}
\end{align}
Character orthogonality also gives
\begin{equation}
  \Pi_{g,c}\Pi_{g,c'}=\delta_{c,c'}\Pi_{g,c},
  \qquad
  \sum_{c\in\F_p}\Pi_{g,c}=I.
  \label{eq:appendix-projector-resolution}
\end{equation}
Thus $\{\Pi_{g,c}:c\in\F_p\}$ is the projective measurement of $W(g)$.

\subsection{Joint character sectors}

Let $g_1,\ldots,g_\ell$ be linearly independent and pairwise symplectically
orthogonal.  The Weyl operators $W(g_j)$ commute, and so do all their spectral
projectors.  For $h=(h_1,\ldots,h_\ell)^T\in\F_p^\ell$, define
\begin{equation}
  \Pi_{G,h}:=\prod_{j=1}^{\ell}\Pi_{g_j,h_j}.
  \label{eq:appendix-joint-projector}
\end{equation}
This is an orthogonal projector.  Moreover,
\begin{equation}
  W(g_j)\Pi_{G,h}=\omega^{h_j}\Pi_{G,h}
  \qquad (1\leq j\leq\ell).
  \label{eq:appendix-joint-character}
\end{equation}
The projectors $\{\Pi_{G,h}:h\in\F_p^\ell\}$ resolve the identity and define
the joint spectral measurement of the commuting family.

For the protocol, $G=G_{Y,Z}$ has top block $A_Y$.  Since $A_Y$ has full column
rank, the columns of $G$ are linearly independent.  The classical isotropy
check
\begin{equation}
  Z^TA_Y-A_Y^TZ=0
  \label{eq:appendix-isotropy-check}
\end{equation}
ensures pairwise symplectic orthogonality, so
$\Pi_{G,h}=\Pi_{Y,Z,h}$ is a well-defined verification projector.

\subsection{Sequential implementation}

Because the projectors commute, the joint test can equivalently be implemented
by measuring the commuting Weyl observables sequentially and accepting only if
all outcomes equal the entries of $h$.  The protocol definition is stated in
terms of the single joint projector because that form is independent of a
particular measurement circuit and is convenient for the later security
analysis.  No claim about fault tolerance or experimental noise is made here.



\begin{thebibliography}{99}

\bibitem{Rompel1990}
J.~Rompel, ``One-Way Functions Are Necessary and Sufficient for Secure
Signatures,'' in \emph{Proceedings of the Twenty-Second Annual ACM Symposium
on Theory of Computing} (ACM, 1990), pp.~387--394.

\bibitem{Shor1994}
P.~W.~Shor, ``Algorithms for Quantum Computation: Discrete Logarithms and
Factoring,'' in \emph{Proceedings of the 35th Annual Symposium on Foundations
of Computer Science} (IEEE, 1994), pp.~124--134.

\bibitem{BernsteinLange2017}
D.~J.~Bernstein and T.~Lange, ``Post-Quantum Cryptography,''
\emph{Nature} \textbf{549}, 188--194 (2017).

\bibitem{NISTFIPS204}
National Institute of Standards and Technology, \emph{Module-Lattice-Based
Digital Signature Standard}, FIPS 204 (2024).

\bibitem{NISTFIPS205}
National Institute of Standards and Technology, \emph{Stateless Hash-Based
Digital Signature Standard}, FIPS 205 (2024).

\bibitem{GottesmanChuang2001}
D.~Gottesman and I.~L.~Chuang, ``Quantum Digital Signatures,''
arXiv:quant-ph/0105032 (2001).

\bibitem{Yin2017}
H.-L.~Yin, Y.~Fu, H.~Liu, Q.-J.~Tang, J.~Wang, L.-X.~You, W.-J.~Zhang,
S.-J.~Chen, Z.~Wang, Q.~Zhang, X.-Y.~Zhou, G.-C.~Guo, and Z.-F.~Han,
``Experimental Quantum Digital Signature over 102 km,''
\emph{Phys. Rev. A} \textbf{95}, 032334 (2017).

\bibitem{Ding2020}
H.-J.~Ding, J.-J.~Chen, L.~Ji, X.-Y.~Zhou, C.-H.~Zhang, C.-M.~Zhang, and
Q.~Wang, ``280-km Experimental Demonstration of a Quantum Digital Signature
with One Decoy State,'' \emph{Opt. Lett.} \textbf{45}, 1711--1714 (2020).

\bibitem{An2019}
X.-B.~An, H.~Zhang, C.-M.~Zhang, W.~Wang, Y.~Wang, and C.-H.~Zhang,
``Practical Quantum Digital Signature with a Gigahertz BB84 Quantum Key
Distribution System,'' \emph{Opt. Lett.} \textbf{44}, 139--142 (2019).

\bibitem{Puthoor2016}
I.~V.~Puthoor, R.~Amiri, P.~Wallden, M.~Curty, and E.~Andersson,
``Measurement-Device-Independent Quantum Digital Signatures,''
\emph{Phys. Rev. A} \textbf{94}, 022328 (2016).

\bibitem{Roberts2017}
G.~L.~Roberts, M.~Lucamarini, Z.~L.~Yuan, J.~F.~Dynes, L.~C.~Comandar,
A.~W.~Sharpe, A.~J.~Shields, M.~Curty, I.~V.~Puthoor, and E.~Andersson,
``Experimental Measurement-Device-Independent Quantum Digital Signatures,''
\emph{Nat. Commun.} \textbf{8}, 1098 (2017).

\bibitem{Richter2021}
S.~Richter, M.~Thornton, I.~Khan, H.~Scott, K.~Jaksch, U.~Vogl, B.~Stiller,
G.~Leuchs, C.~Marquardt, and N.~Korolkova, ``Agile and Versatile Quantum
Communication: Signatures and Secrets,'' \emph{Phys. Rev. X} \textbf{11},
011038 (2021).

\bibitem{Lu2021}
Y.-S.~Lu, X.-Y.~Cao, C.-X.~Weng, J.~Gu, Y.-F.~Huang, H.-L.~Yin,
Z.-B.~Chen, and G.-C.~Guo, ``Efficient Quantum Digital Signatures without
Symmetrization Step,'' \emph{Opt. Express} \textbf{29}, 10162--10171 (2021).

\bibitem{Qin2024}
J.-Q.~Qin, Z.-W.~Yu, and X.-B.~Wang, ``Efficient Quantum Digital Signatures
over Long Distances with Likely Bit Strings,'' \emph{Phys. Rev. Applied}
\textbf{21}, 024012 (2024).

\bibitem{Yin2023}
H.-L.~Yin, Y.~Fu, C.-L.~Li, C.-X.~Weng, B.-H.~Li, J.~Gu, Y.-S.~Lu,
S.~Huang, and Z.-B.~Chen, ``Experimental Quantum Secure Network with Digital
Signatures and Encryption,'' \emph{Natl. Sci. Rev.} \textbf{10}, nwac228
(2023).

\bibitem{Du2025}
Y.~Du, B.-H.~Li, X.~Hua, X.-Y.~Cao, Z.~Zhao, F.~Xie, Z.~Zhang, H.-L.~Yin,
X.~Xiao, and K.~Wei, ``Chip-Integrated Quantum Signature Network over 200 km,''
\emph{Light Sci. Appl.} \textbf{14}, 108 (2025).

\bibitem{MorimaeYamakawa2022}
T.~Morimae and T.~Yamakawa, ``Quantum Commitments and Signatures without
One-Way Functions,'' in \emph{Advances in Cryptology---CRYPTO 2022}
(Springer, 2022), pp.~269--295.

\bibitem{MorimaeYamakawa2024}
T.~Morimae and T.~Yamakawa, ``One-Wayness in Quantum Cryptography,'' in
\emph{19th Conference on the Theory of Quantum Computation, Communication and
Cryptography (TQC 2024)}, LIPIcs \textbf{310}, Article 4 (2024).

\bibitem{KhuranaTomer2024}
D.~Khurana and K.~Tomer, ``Commitments from Quantum One-Wayness,''
arXiv:2310.11526 (2024).

\bibitem{ColadangeloMutreja2024}
A.~Coladangelo and S.~Mutreja, ``On Black-Box Separations of Quantum Digital
Signatures from Pseudorandom States,'' arXiv:2402.08194 (2024).

\bibitem{MorimaePorembaYamakawa2024}
T.~Morimae, A.~Poremba, and T.~Yamakawa, ``Revocable Quantum Digital
Signatures,'' in \emph{19th Conference on the Theory of Quantum Computation,
Communication and Cryptography (TQC 2024)}, LIPIcs \textbf{310}, Article 5 (2024).

\bibitem{WangHayashi2026QDS}
W.~Wang and M.~Hayashi, ``Quantum Digital Signature Based on Single-Qubit
without a Trusted Third-Party,'' \emph{Adv. Quantum Technol.} \textbf{9},
e00828 (2026).

\bibitem{KawachiEtAl2012}
A.~Kawachi, T.~Koshiba, H.~Nishimura, and T.~Yamakami, ``Computational
Indistinguishability between Quantum States and Its Cryptographic
Application,'' \emph{J. Cryptol.} \textbf{25}, 528--555 (2012).

\bibitem{HayashiKawachiKobayashi2008}
M.~Hayashi, A.~Kawachi, and H.~Kobayashi, ``Quantum Measurements for Hidden
Subgroup Problems with Optimal Sample Complexity,'' \emph{Quantum Inf.
Comput.} \textbf{8}, 345--358 (2008).

\bibitem{ArapinisEtAl2021}
M.~Arapinis, M.~Delavar, M.~Doosti, and E.~Kashefi, ``Quantum Physical
Unclonable Functions: Possibilities and Impossibilities,'' \emph{Quantum}
\textbf{5}, 475 (2021).

\bibitem{GhoshEtAl2024}
S.~Ghosh, V.~Galetsky, P.~Juli\`a Farr\'e, C.~Deppe, R.~Ferrara, and H.~Boche,
``Existential Unforgeability in Quantum Authentication from Quantum Physical
Unclonable Functions Based on Random von Neumann Measurement,''
\emph{Phys. Rev. Research} \textbf{6}, 043306 (2024).

\bibitem{HayashiQuadraticStabilizerPrediction}
M.~Hayashi and Y. Lu, ``Partial Stabilizer Learning under Fixed Commuting Constraints and the Absence of a Copy-Rate Discount,''
arXiv:2609.13923 (2026).

\end{thebibliography}
\end{document}